\documentclass[sigplan,screen]{acmart}
\newcommand{\oomit}[1]{}
\usepackage[T1]{fontenc}
\usepackage{amsmath,amsfonts}
\usepackage{bm}
\usepackage{textcomp}
\usepackage{cases}
\usepackage{empheq}
\usepackage{cancel,soul,ulem}
\usepackage{subfigure}
\usepackage{ulem}
\usepackage{marvosym}
\newtheorem{remark}{Remark}
\newtheorem{definition}{Definition}
\newtheorem{theorem}{Theorem}
\newtheorem{proposition}{Proposition}
\newtheorem{lemma}{Lemma}
\newtheorem{example}{Example}

\newtheorem{assumption}{Assumption}
\newtheorem{problem}{Problem}
\usepackage{multirow}
\usepackage{caption}

\AtBeginDocument{%
  }

\setcopyright{acmlicensed}
\copyrightyear{2026}
\acmYear{2026}
\acmDOI{XXXXXXX.XXXXXXX}

\acmConference[HSCC '26]{29th ACM International Conference on Hybrid Systems: Computation and Control}{May 11--14,2026}{Saint Malo, France}

\acmISBN{978-1-4503-XXXX-X/2018/06}

\begin{document}

\pagenumbering{arabic}
\setcounter{page}{1}
\title{PAC Finite-Time Safety Guarantees for Stochastic Systems with Unknown Disturbance Distributions}



\author{Taoran Wu}
\affiliation{%
  \institution{KLSS, Institute of Software, CAS $\&$ University of}
  \city{Chinese Academy of Sciences, Beijing}
  \country{China}}
\email{wutr@ios.ac.cn}

\author{Dominik Wagner}
\affiliation{%
  \institution{College of Computing and Data Science, Nanyang Technological University,Singapore}
  \city{}
  \country{}}
\email{dominik.wagner@ntu.edu.sg}

\author{C.-H. Luke Ong\textsuperscript{\Letter}}
\affiliation{%
  \institution{College of Computing and Data Science, Nanyang Technological University, Singapore}
  \city{}
  \country{}}
\email{luke.ong@ntu.edu.sg}

\author{Bai Xue\textsuperscript{\Letter}}
\affiliation{%
  \institution{KLSS, Institute of Software, CAS $\&$ University of}
  \city{Chinese Academy of Sciences, Beijing}
  \country{China}}
\email{xuebai@ios.ac.cn}



\begin{abstract}
We investigate the problem of establishing finite-time probabilistic safety guarantees for discrete-time stochastic dynamical systems subject to unknown disturbance distributions, using barrier certificate methods. Our approach develops a data-driven safety certification framework that relies only on a finite collection of independent and identically distributed (i.i.d.) disturbance samples. Within this framework, we propose a certification procedure such that, with confidence at least $1-\delta$ over the sampled disturbances, if the output of the certification procedure is accepted, the probability that the system remains within a prescribed safe set over a finite horizon is at least $1-\epsilon$. A key challenge lies in formally characterizing the probably approximately correct (PAC) generalization behavior induced by finite samples. To address this, we derive PAC generalization bounds using tools from  VC dimension, scenario optimization, and Rademacher complexity. These results illuminate the fundamental trade-offs between sample size, model complexity, and safety tolerance, providing both theoretical insight and practical guidance for designing reliable, data-driven safety certificates in discrete-time stochastic systems.
\end{abstract}



\maketitle

\section{Introduction}

Stochastic systems are widely used in safety-critical applications because they provide a principled framework for modeling and reasoning about inherent uncertainties \cite{bertsekas1996stochastic}, particularly aleatoric uncertainty arising from natural randomness in the system or environment \cite{sullivan2015introduction,thiebes2021trustworthy}. In real-world settings such as autonomous vehicles, robotics, and aerospace, system behavior is often influenced by unpredictable disturbances, sensor noise, or environmental variability. By explicitly accounting for these stochastic effects, safety verification can provide probabilistic guarantees of system performance, robustness, and safety—an essential requirement for deploying reliable systems in high-stakes scenarios \cite{baier2008principles}.


Traditional approaches to stochastic safety often rely on complete knowledge of the disturbance distribution or adopt worst-case robust methods. Distribution-based approaches can fail when the assumed distribution deviates from reality, whereas robust methods are often overly conservative, limiting system performance. In many practical settings involving hybrid and cyber-physical systems, the true disturbance distribution is unknown, and only finite samples of disturbances are available from sensors, simulations, or historical data. Although several recent works have begun to address safety guarantees for such systems, e.g., \cite{mathiesen2023inner,mathiesen2024data}, existing methods remain at an early stage of development. 

In this work, we develop a data-driven framework to certify PAC finite-time safety guarantees for discrete-time stochastic systems subject to unknown disturbance distributions via the computation of barrier functions. 
This framework relies only on a finite collection of independent and identically distributed (i.i.d.) disturbance samples and  provides a certification procedure such that, with confidence at least $1-\delta$ over the sampled disturbances, if the output of the certification procedure is accepted, the probability that the system remains within a prescribed safe set over a finite horizon is at least $1-\epsilon$.
A key technical challenge lies in characterizing the generalization gap between the empirical safety estimate obtained from finite samples and the true lower bound of the finite-time safety probability. To address this issue, 
we derive PAC-style generalization bounds using tools from statistical learning theory, including VC dimension, scenario optimization, and Rademacher complexity. These bounds are explicit and yield concrete sample-complexity requirements that relate the number of disturbance samples to the complexity of the barrier functions, as well as to $\epsilon$ and $\delta$. This enables practitioners to determine the number of samples required to achieve a desired confidence level and safety guarantee. Finally, the proposed methods were illustrated on several numerical examples using a semidefinite programming solver.

The core novelty of this work lies in their architectural integration (specifically, the strategy of constructing PAC bounds for barrier certificates) to yield safety guarantees unavailable in prior data-driven methods. The main contributions of this work are listed below:
\begin{enumerate}
    \item \textbf{PAC Analysis of Sample-Based Barrier Certificates:}  
    We derive VC-, scenario-, and Rademacher-based generalization bounds for barrier certificates constrained using disturbance samples, providing formal safety guarantees beyond the observed data.

    \item \textbf{Applicability to General Nonlinear Dynamics:}  
    Unlike prior data-driven approaches that are restricted to piecewise-affine systems~\cite{gracia2024data,mathiesen2024data}, our framework applies to general nonlinear stochastic dynamics.

    \item \textbf{Unified View of Generalization Guarantees:}  
    We systematically derive and compare multiple PAC bounds within a single certification framework, clarifying trade-offs between conservatism and sample efficiency.
\end{enumerate}



\subsection*{Related Work}

Barrier certificates were first introduced for deterministic systems as a Lyapunov-like tool for verifying safety and reachability properties \cite{prajna2004safety,prajna2007convex}, and were later extended to stochastic systems for infinite-horizon safety and reach-avoid verification. Notable contributions include establishing lower bounds on safety and reach-avoid probabilities via Ville’s inequality \cite{prajna2007framework,vzikelic2023learning,zhi2024unifying}, as well as deriving both lower and upper bounds on reach-avoid and safety probabilities through equation relaxations \cite{xue2021reach,xue2022reach,yu2023safe,xue2024sufficient,chen2025construction}. Extensions also cover infinite-horizon probabilistic program analysis and $\omega$-regular properties \cite{chakarov2013probabilistic,mciver2017new,kenyon2021supermartingales,abate2024stochastic,wang2025verifying}. However, in practice, finite-time verification is more relevant, as most real systems operate under bounded time constraints. $c$-martingale–based barrier conditions have become an important tool for the finite-time formal verification of stochastic systems, with early foundations laid in \cite{kushner1967,steinhardt2012finite}. These methods were later extended to finite-time temporal-logic verification and controller synthesis under robust-invariance or stopped-process assumptions \cite{jagtap2018temporal,jagtap2020formal,santoyo2021barrier}. More recently, \cite{zhi2024unifying,xue2024finite} extended $c$-martingale–based barrier conditions and refined the associated safety probability bounds, while \cite{ghanbarpour2025characterization} investigated the safety verification problem for stochastic discrete-time systems modeled by difference inclusions. 
The above methods assume full knowledge of the system, including both dynamics and disturbance distributions. In contrast, this work considers systems with known dynamics but unknown disturbances, reflecting practical scenarios—such as autonomous systems operating in uncertain environments—where exact disturbance statistics are unavailable.

On the other hand, data-driven methods for formal verification and design of black-box systems have received increasing attention, e.g., \cite{samari2024single,wooding2024learning}. A representative data-driven approach is scenario optimization, originally proposed in \cite{calafiore2006scenario} to address robust optimization problems. 
It has since been integrated with barrier certificates for safety verification and controller synthesis in deterministic systems \cite{xue2019probably,nejati2023formal, wu2025convex, rickard2025data, samari2025data}. Later, \cite{nejati2023data,salamati2024data} extended scenario optimization to the stochastic setting. \cite{salamati2022data} and \cite{salamati2022safety} further reduced sample complexity using the wait-and-judge technique \cite{campi2018wait} and the repetitive scenario technique \cite{calafiore2016repetitive}, respectively. All these works on stochastic systems require prior knowledge of certain constants, such as the Lipschitz constants of the system dynamics and the barrier certificates. 
\textcolor{black}{In contrast, our work focuses on a complementary setting in which the system dynamics are assumed to be known, while the disturbance distribution is unknown. We address this uncertainty using only i.i.d. disturbance samples, without imposing additional structural assumptions, such as Lipschitz continuity on the system dynamics.} \textcolor{black}{Moreover, we would like to highlight a technical subtlety in these approaches. To the best of our knowledge, works such as \cite{salamati2024data,salamati2022data,salamati2022safety} apply Chebyshev’s inequality to obtain uniform characterizations of expectations (e.g., via Rademacher complexity). However, the validity of this step in the current setting requires additional justification, since applying a pointwise bound to a data-dependent selection does not automatically extend to a uniform guarantee.}
This setting is closely related to \cite{mathiesen2023inner,gracia2024data,mathiesen2024data}. In \cite{mathiesen2023inner}, the authors studied finite-horizon safety guarantees for discrete-time piecewise affine stochastic systems with unknown disturbance distributions. Their approach uses stochastic barrier certificates to reformulate the problem as a chance-constrained optimization problem and synthesizes piecewise affine stochastic barrier functions via scenario optimization. Because it relies on this chance-constrained formulation, the method is conceptually similar to the robust barrier certificate approach proposed in this work, which enforces constraint satisfaction over all sampled disturbances. More recently, \cite{mathiesen2024data} extended this method to nonlinear systems with affine disturbances by over-approximating them as uncertain piecewise affine systems. In contrast, our method is applicable to more general system classes without requiring piecewise affine approximations or restricting the barrier function structure to (piecewise) affine forms. Moreover, to reduce the computational overhead associated with solving the resulting optimization problem, we further propose an alternative approach that synthesizes stochastic barrier functions directly, without performing a transformation into a chance-constrained optimization problem as in \cite{mathiesen2023inner,mathiesen2024data}. For a detailed discussion of the relations between \cite{mathiesen2023inner,mathiesen2024data} and the present work, please refer to Remark~ \ref{relation}. Further, \cite{gracia2024data} addressed the controller synthesis for nonlinear switched systems with unknown disturbance distributions under LTL$_f$ specifications. Their method is based on learning an ambiguity set that contains the unknown disturbance distribution and subsequently constructing a robust Markov decision process as a finite abstraction of the system. However, it is known that such discretization-based approaches suffer from the curse of dimensionality.

\paragraph{Organization.} Section~\ref{sec:pre} introduces the system model, notation, and the PAC finite-time safety certification problem. Section~\ref{sec:pac} presents the main PAC guarantees and explicit sample-complexity bounds using VC
dimension, scenario optimization, and Rademacher complexity.  Section~\ref{sec:exam} demonstrates the effectiveness of the proposed approaches, and Section \ref{sec:con} concludes this paper.

\paragraph{Notation.}
We use $\mathbb{R}^n$ to denote the $n$-dimensional Euclidean space, and $\mathbb{N}$ for the set of nonnegative integers. For two sets $\mathbb{A}$ and $\mathbb{B}$, $\mathbb{A} \setminus \mathbb{B}$ denotes the set difference, i.e., the set of all elements that belong to $\mathbb{A}$ but not to $\mathbb{B}$. The notation $\overline{\mathbb{A}}$ denotes the closure of $\mathbb{A}$, $\mathbb{A}^c$ denotes its complement, and $\mathbb{A}^{\circ}$ denotes its interior. Given a probability space $(\Omega, \mathcal{F}, \textnormal{P}_{\bm{y}})$, we denote by 
$\textnormal{P}_{\bm{y}}[\mathbb{E}]$ the probability of an event 
$\mathbb{E} \in \mathcal{F}$, and by 
$\textnormal{E}_{\bm{y}}[\cdot]$ the expectation with respect to $\textnormal{P}_{\bm{y}}$.

\section{Preliminaries}
\label{sec:pre}
In this section, we introduce stochastic systems and the PAC finite-time safety certification problem of interest.

\paragraph{System model.}
We consider a discrete-time stochastic system
\begin{equation}
    \bm{x}(j+1) = \bm{f}(\bm{x}(j), \bm{d}(j)), \quad j \in \mathbb{N},
    \label{eq:system}
\end{equation}
where $\bm{x}(j) \in \mathbb{R}^n$ is the system state and $\bm{d}(j) \in \mathbb{D}$ is a stochastic disturbance with $\mathbb{D}$ being compact in $\mathbb{R}^{d}$. The update map $\bm{f}(\cdot,\cdot): \mathbb{R}^n \times \mathbb{D} \to \mathbb{R}^n$ is assumed known,  deterministic, and continuous over $(\bm{x},\bm{d})$. The disturbance sequence \(\{\bm{d}(j)\}_{j\in \mathbb{N}}\) consists of i.i.d. random vectors on a probability space \((\mathbb{D},\mathcal{F},\text{P}_{\bm{d}})\), with \(\bm{d}(j)\sim \text{P}_{\bm{d}}\) for all \(j\in \mathbb{N}\). The expectation \(\textnormal{E}_{\bm{d}}[\cdot]\) is taken with respect to \(\text{P}_{\bm{d}}\). 

\begin{definition}[Disturbance Signal]  
A disturbance signal \(\pi\) is a sample path of the stochastic process \(\{\bm{d}(j) \colon \mathbb{D} \rightarrow \mathbb{D}, j \in \mathbb{N}\}\), defined on the canonical sample space \(\mathbb{D}^{\infty}\) with the product topology \(\mathcal{B}(\mathbb{D}^{\infty})\) and probability measure \(\textnormal{P}_{\pi}:=\textnormal{P}_{\bm{d}}^{\infty}\).   
\end{definition}  

\begin{definition}[Trajectory]  
Given an initial state \(\bm{x}_0 \in \mathbb{R}^n\) and a disturbance signal \(\pi\), the corresponding trajectory of system \eqref{eq:system} is the mapping \(\bm{\phi}_{\pi}^{\bm{x}_0}(\cdot) \colon \mathbb{N} \rightarrow \mathbb{R}^n\) defined recursively as  
\[
\begin{cases}
    \bm{\phi}_{\pi}^{\bm{x}_0}(j+1) = \bm{f}(\bm{\phi}_{\pi}^{\bm{x}_0}(j), \bm{d}(j)), \quad \forall j\in \mathbb{N},\\
    \bm{\phi}_{\pi}^{\bm{x}_0}(0) = \bm{x}_0.
\end{cases}
\]  
\end{definition}

\paragraph{$k$-step safety certification.}
A compact set $\mathbb{X} \subseteq \mathbb{R}^n$ defines the safe region. Given state $\bm{x} \in \mathbb{X}$, the system starting from $\bm{x}\in \mathbb{X}$ is $k$-step safe with respect to $\epsilon(\cdot): \mathbb{X}\rightarrow [0,1]$  if
\[
    \text{P}_{\pi}\big[ \wedge_{j=1}^k \bm{\phi}^{\bm{x}}_{\pi}(j) \in \mathbb{X} \big] \geq 1-\epsilon(\bm{x}),
\]
where $k\in \mathbb{N}$.

In this paper, instead of knowing $\text{P}_{\bm{d}}$, we assume access to a finite i.i.d. sample set
\[
    D^M = \{ \bm{d}^{(1)}, \dots, \bm{d}^{(M)} \} \sim \text{P}_{\bm{d}}^M,
\]
where $\text{P}_{\bm{d}}^M$ denotes the $M$-fold product distribution of $\text{P}_{\bm{d}}$, and each sample $\bm{d}^{(i)} \in \mathbb{D}$ is drawn independently according to $\text{P}_{\bm{d}}$ for $i=1,\dots,M$.

\begin{problem}[PAC Finite-time Safety Certification]
\label{prob:1}
For a given finite time horizon $k \in \mathbb{N}$, the goal is to design a certification procedure $\mathcal{A}$ that provides k-step safety certification in the PAC sense: for each $\bm{x}\in \mathbb{X}$, with confidence at least $1-\delta$ over the sample set $D^M$, if $\mathcal{A}(D^M)$ accepts, then the system \eqref{eq:system} starting from $\bm{x}$ is $k$-step safe with respect to $\epsilon(\bm{x})$, where $\epsilon(\cdot): \mathbb{X} \rightarrow [0,1]$ specifies a state-dependent probability threshold that may depend on $D^M$. 

Formally, the $k$-step PAC guarantee can be expressed as
\[
\textnormal{P}_{\bm{d}}^M \Bigg[
\mathcal{A}(D^M) \Rightarrow \left(
\begin{split}
&\textnormal{P}_{\pi}\Big[\bigwedge_{j=1}^k \bm{\phi}^{\bm{x}}_{\pi}(j) \in \mathbb{X} \Big] \\
&\ge 1-\epsilon(\bm{x})
\end{split}
\right)
\Bigg] \ge 1-\delta, \forall \bm{x}\in \mathbb{X}.
\]
\end{problem}

In addition, in the sequel, we denote by $\widehat{\mathbb{X}}$ a compact set that contains $\mathbb{X}$ as well as all possible states reachable by system \eqref{eq:system} from $\mathbb{X}$ within one step, i.e., 
    \[\widehat{\mathbb{X}} \supseteq \{\bm{y}\mid \bm{y}=\bm{f}(\bm{x},\bm{d}), \bm{x}\in \mathbb{X},\bm{d}\in \mathbb{D}\} \cup \mathbb{X}.\]
    Since $\bm{f}(\bm{x},\bm{d})$ is continuous over $(\bm{x},\bm{d})$ and both $\mathbb{X}$ and $\mathbb{D}$ are compact, such a compact set $\widehat{\mathbb{X}}$ always exists.

\section{PAC Finite-time Safety Certification}
\label{sec:pac}
In this section, we propose data-driven approaches to address Problem \ref{prob:1} by combining robust and stochastic barrier certificates with tools from statistical learning theory, including VC dimension, scenario approach, and Rademacher complexity. The integration of robust barrier certificates with VC dimension and scenario approach is presented in Subsection \ref{sec:robust}, while the integration of stochastic barrier certificates with Rademacher complexity is discussed in Subsection \ref{sec:probabilistic}.

\subsection{Certification via Robust Barrier Functions}
\label{sec:robust}
In this subsection, we present approaches that integrate VC dimension and scenario approach with robust barrier functions (RBFs) to address Problem \ref{prob:1},
beginning with a review of RBFs for discrete-time systems.

\begin{definition}
    Given $\gamma \in (0,1)$ and a safe set $\mathbb{X}$, a function $h(\cdot):\mathbb{R}^n\rightarrow\mathbb{R}$ is a RBF for the system 
    \eqref{eq:system} if the inequalities hold:
           \begin{subequations}
    \label{eq:rbf}
    \begin{empheq}[left=\empheqlbrace]{align}
        &h(\bm{x}) \leq  0, & \forall \bm{x} \in \overline{\widehat{\mathbb{X}}\setminus\mathbb{X}}, \label{eq:rbf1} \\
     &h(\bm{f}(\bm{x},\bm{d})) \geq \gamma h(\bm{x}), & \forall \bm{x} \in \mathbb{X}, ~\forall \bm{d} \in \mathbb{D}, \label{eq:rbf3}
    \end{empheq}
    \end{subequations}
    and \textcolor{black}{~$\mathbb{X}_s=\{\bm{x}\in \mathbb{X}\mid h(\bm{x})> 0\}\neq \emptyset$}.
\end{definition}

\begin{proposition}[\cite{prajna2007framework}]
    If there exists a RBF $h(\cdot): \mathbb{R}^n\rightarrow\mathbb{R}$, then for any disturbance $\bm{d} \in \mathbb{D}$ and initial state $\bm{x}_0 \in \mathbb{X}_s=\{\bm{x}\in \mathbb{X}\mid h(\bm{x})> 0\}$, the system \eqref{eq:system} will remain within $\mathbb{X}_s$—and consequently within the safe set $\mathbb{X}$—at the next time step, i.e., $\forall \bm{d} \in \mathbb{D}, \forall \bm{x}_0\in \mathbb{X}_s, \bm{f}(\bm{x}_0,\bm{d}) \in \mathbb{X}_s$.
\end{proposition}


To search for such a function, we consider a parameterized function $h(\bm{a}, \cdot): \mathbb{R}^n \rightarrow \mathbb{R}$ with $h(\bm{0},\bm{x})\equiv 0$ and aim to find the parameter vector $\bm{a} = (\bm{a}[1], \ldots, \bm{a}[m])^{\top} \in \mathbb{R}^m$ such that $h(\bm{a}, \bm{x})$ satisfies the constraints \eqref{eq:rbf1} and \eqref{eq:rbf3}. This search problem can be reformulated as the following uncertain optimization one:
\begin{equation}
\begin{split}
    &\min_{\bm{a}\in \mathbb{R}^m } - \textstyle\frac{1}{N_o}\textstyle\sum_{i=1}^{N_o}h(\bm{a},\bm{x}'_i)\label{eq:RBF_linear1} \\
    \text{s.t.}&\begin{cases}
        h(\bm{a},\bm{x})\leq 0, \hspace{2.6cm} \forall \bm{x}\in \overline{\widehat{\mathbb{X}}\setminus \mathbb{X}},\\
        h(\bm{a}, \bm{f}(\bm{x},\bm{d})) \geq \gamma h(\bm{a}, \bm{x}), \hspace{0.5cm}\forall \bm{x} \in \mathbb{X}, \forall \bm{d} \in \mathbb{D},\\
        \bm{a}[l]  \in [-U_a, U_a], ~l=1,\ldots,m, 
    \end{cases} 
\end{split}
\end{equation}
where $\gamma \in (0,1)$ and $U_a > 0$ are user-defined scalar values, and $\{\bm{x}'_i\}_{i=1}^{N_o}$ is a family of specified states distributed evenly over $\mathbb{X}$, which can be generated by random sampling from a uniform distribution over the safe set $\mathbb{X}$.

\begin{remark}
   The ideal objective of the program \eqref{eq:RBF_linear1} is to find a function $h(\bm{a}, \bm{x})$ that is maximized over $\mathbb{X}$, thereby providing the largest possible $\mathbb{X}_s$. A natural choice for the objective function is to maximize $\int_{\mathbb{X}} h(\bm{a}, \bm{x}) d\bm{x}$. However, since this integral may not be analytically tractable, we approximate it using numerical sampling $\frac{1}{N_o}\sum_{i=1}^{N_o}h(\bm{a},\bm{x}'_i)$. 
\end{remark}



The optimization problem \eqref{eq:RBF_linear1} always admits a feasible solution, since $\bm{a}\equiv \bm{0}$ is one such solution. Let $\bm{a}^*$ be the solution to \eqref{eq:RBF_linear1}. If~\textcolor{black}{~$\mathbb{X}_s=\{\bm{x}\in \mathbb{X}\mid h(\bm{a}^*,\bm{x})> 0\}\neq \emptyset$}, we conclude that $h(\bm{a}^*, \bm{x})$ is a RBF. However, problem \eqref{eq:RBF_linear1}, which enforces worst-case satisfaction over all possible disturbances, can be overly conservative or even infeasible in practice. Such conservatism may result in degraded performance or the failure to identify a feasible safety certificate. To alleviate this issue, we adopt a probabilistic formulation of the safety certification problem, which provides quantitative safety guarantees instead of relying on excessively conservative worst-case bounds. Nevertheless, the lack of knowledge about the true disturbance distribution prevents the direct computation of such probabilistic characterizations. To address this challenge, we resort to a data-driven approach, leveraging finite disturbance samples to obtain confidence-guaranteed probabilistic safety assurances. Specially, we draw $M$ i.i.d disturbance samples of the form $D^M=\{\bm{d}^{(i)}\}_{i=1}^{M}$, where each $\bm{d}^{(i)}$ is independently drawn from $\mathbb{D}$ according to $\textnormal{P}_{\bm{d}}$. We then replace the universal disturbance constraints in \eqref{eq:RBF_linear1} with their sample-based counterparts, leading to the following sample-based optimization problem:

\begin{equation}
\begin{split}
    &\min_{\bm{a}\in \mathbb{R}^m } - \textstyle\frac{1}{N_o}\textstyle\sum_{i=1}^{N_o}h(\bm{a},\bm{x}'_i)\label{eq:RBF_linear21} \\
    \text{s.t.}&\begin{cases}
     h(\bm{a},\bm{x})\leq 0, \hspace{3.15cm}\forall \bm{x}\in \overline{\widehat{\mathbb{X}}\setminus \mathbb{X}},\\
        h(\bm{a}, \bm{f}(\bm{x},\bm{d}^{(j)})) \geq \gamma h(\bm{a}, \bm{x}), \hspace{1.2cm}\forall \bm{x}\in \mathbb{X},  \\
        \bm{a}[l]  \in [-U_a, U_a];~ j=1,\ldots,M;~ l=1,\ldots,m. \\
    \end{cases}
\end{split}
\end{equation}

Similar to program \eqref{eq:RBF_linear1}, program \eqref{eq:RBF_linear21} is always feasible.

\subsubsection{VC Dimension Based Methods}
In this subsection, we employ the VC dimension to derive uniform generalization bounds over all sample disturbances for the solution to \eqref{eq:RBF_linear21}, thereby establishing a PAC-style characterization of one-step (i.e., $k=1$) safety certification.

 The gain of PAC characterization of one-step safety certification relies on the following assumption. 
\begin{assumption}[Uniformly Finite VC dimension]
    \label{ass:VC_finite}
   Define the function class
\[
\mathbb{F}_{\bm{x}}
= \left\{\widehat{f}_{\bm{a}}(\bm{d})\!=\!\mathbf 1_{\left\{\bm{d}\mid h(\bm{a},\bm{f}(\bm{x},\bm{d}))-\gamma h(\bm{a},\bm{x})<0\right\}}(\bm{d})\mid\bm{a}\!\in\![-U_a,U_a]^m\right\}.
\]
Assume $\mathbb{F}_{\bm{x}}$ has VC dimension at most $N$ for every $\bm{x}\in \mathbb{X}$ (uniformly in $\bm{x}$). 
\end{assumption}

When $h(\bm{a},\bm{x})$ is linear in $\bm{a}$, we can write $h(\bm{a},\bm{y})=\bm{a}^\top \bm{g}(\bm{y})$. Thus, the indicator function $\widehat{f} \in \mathbb{F}_{\bm{x}}$ is characterized as:
\[
\widehat{f}_{\bm{a}}(\bm{d})=\mathbf 1_{\left\{\bm{d}\mid \bm{a}^\top\bm{\psi}(\bm{x},\bm{d})<0\right\}}(\bm{d}), \bm{\psi}(\bm{x},\bm{d}):=\bm{g}(\bm{f}(\bm{x},\bm{d}))-\gamma \bm{g}(\bm{x}).
\]
This defined a linear family in the parameters $\bm{a}$, hence its VC dimension is less than or equal to $m$.

Let $\textnormal{E}_{\bm{d},M}$ denote the empirical mean of $\widehat{f}_{\bm{a}}(\bm{d})$ evaluated over the set of $M$ i.i.d.\ disturbance samples $D^M$, i.e.,
\[
\textstyle\textnormal{E}_{\bm{d},M}[\widehat{f}_{\bm{a}}] := \frac{1}{M} \sum_{i=1}^M \widehat{f}_{\bm{a}}(\bm{d}^{(i)}), \forall \bm{x}\in \mathbb{X}.
\]

Since there always exist solutions to \eqref{eq:RBF_linear21}, there always exists a solution \textcolor{black}{$\bm{a}(D^M)$} such that $\textnormal{E}_{\bm{d},M}[\widehat{f}_{\bm{a}(D^M)}]=0, \forall \bm{x}\in \mathbb{X}$. 
According to Corollary~4 in~\cite{alamo2009randomized}, there exists an explicit bound on the number of samples $M$ required to ensure that the empirical probability $\textnormal{E}_{\bm{d},M}[\widehat{f}_{\bm{a}(D^M)}]$ deviates from the true probability $\textnormal{E}_{\bm{d}}[\widehat{f}_{\bm{a}(D^M)}]$ by at most $\epsilon$ from one side with confidence at least $1-\delta$. Specifically,
\begin{equation}
\label{MM}
M \ge \frac{5}{\epsilon} \left( \ln{\frac{4}{\delta}} + N \ln{\frac{40}{\epsilon}} \right)
\end{equation}
is sufficient to guarantee that
\[
\textnormal{P}_{\bm{d}}^M\!\left[\textnormal{E}_{\bm{d}}[\widehat{f}_{\bm{a}(D^M)}]\le \epsilon \right] \ge 1-\delta, \forall \bm{x}\in \mathbb{X}.
\]

\begin{lemma}[Corollary 4, \cite{alamo2009randomized}]
\label{coro4}
Suppose $\epsilon \in (0,1)$ and $\delta\in (0,1)$ are given. Then, for each $\bm{x}\in \mathbb{X}$, the probability of one-sided constrained failure $p(M,\epsilon,0)$ is smaller than $\delta$, where $p(M,\epsilon,0)$ is defined as
\[p(M,\epsilon,0):=\textnormal{P}_{\bm{d}}^M\left[
\begin{split}
&\exists ~\bm{a}\in [-U_a,U_a]^m, \\&\Big(
\textnormal{E}_{\bm{d},M}[\widehat{f}_{\bm{a}}]\leq 0 \wedge \textnormal{E}_{\bm{d}}[\widehat{f}_{\bm{a}}]>\epsilon\Big)
\end{split}
\right], \]  
if $M$ satisfies \eqref{MM}.
\end{lemma}

Under Assumption \ref{ass:VC_finite} and Lemma \ref{coro4}, we obtain the following conclusion. For each $\bm{x}\in\mathbb{X}$, with confidence at least $1-\delta$,
\[
\textnormal{P}_{\bm d}\!\big[h(\bm a^*(D^M),\bm f(\bm x,\bm d))\ge\gamma\,h(\bm a^*(D^M),\bm x)\big] \ge 1-\epsilon,
\]
where $\bm a^*(D^M)$ is obtained by solving \eqref{eq:RBF_linear21}. Consequently, with the same confidence, if $\bm{x}$ belongs to
\[
\mathbb{X}_s = \{\bm{x}\in\mathbb{X} \mid h(\bm a^*(D^M),\bm{x}) > 0\},
\]
the probability that the system \eqref{eq:system} starting from $\bm{x}$ remains within $\mathbb{X}_s$ at the next step is at least $1-\epsilon$. The above statements hold provided that the sample size $M$ satisfies \eqref{MM}.

\begin{theorem}[State-pointwise PAC one-step safety guarantee I]
\label{lemma:robust}
    Let $\epsilon, \delta$, and $\gamma \in (0,1)$ be given constants. Let  $\bm{a}^*(D^M)$ be the (locally/globally) optimal solution to \eqref{eq:RBF_linear21} with $D^M=\{\bm{d}^{(j)}\}_{j=1}^M\stackrel{\text{i.i.d.}}{\sim}\textnormal{P}_{\bm{d}}$, and define the certification procedure
\begin{equation}
\label{cp}
\mathcal{A}(D^M) := 
\begin{cases}
1, & \text{if } ~\textcolor{black}{\mathbb{X}_s\neq \emptyset},\\
0, & \text{otherwise,}
\end{cases}
\end{equation}
where $\bm{a}^*(D^M)$ is obtained from \eqref{eq:RBF_linear21}.
  Then the following probabilistic guarantee holds:
    \begin{equation*}
\begin{split}
  & \textnormal{P}_{\bm{d}}^M \left[{\mathbb{X}_s\neq \emptyset} \Rightarrow 
  \Big( \begin{split}
\textnormal{P}_{\bm{d}}\left[
\begin{split}
&h(\bm{a}^*(D^M),\bm{f}(\bm{x},\bm{d})) \\
& \geq\gamma h(\bm{a}^*(D^M),\bm{x})
\end{split}
\right]\geq 1-\epsilon\Big)
        \end{split}
        \right]\\   
          &\geq 
          1-\delta,\forall \bm{x}\in \mathbb{X}, 
\end{split}
\end{equation*}
implying 
\begin{equation*}
\begin{split}
   \textnormal{P}_{\bm{d}}^M \left[
    \begin{split}
        {\mathbb{X}_s\!\neq\! \emptyset} \Rightarrow \textnormal{P}_{\bm{d}}[\bm{f}(\bm{x},\bm{d}) \!\in\! \mathbb{X}_s]\geq 1-\epsilon(\bm{x})
        \end{split}
        \right]\geq 
          1-\delta, \forall \bm{x}\in \mathbb{X},
\end{split}
\end{equation*}
where $\mathbb{X}_s=\{\bm{x}\in \mathbb{X}\mid h(\bm{a}^*(D^M),\bm{x}) > 0\}\subseteq \mathbb{X}$, \begin{equation}
\label{epsilon}
\epsilon(\bm{x}) := 
\begin{cases}
\epsilon, & \text{if } \bm{x} \in \mathbb{X}_s,\\[1ex]
1, & \text{if } \bm{x} \in \mathbb{X} \setminus \mathbb{X}_s,
\end{cases}
\end{equation}
and the sample number $M$ satisfies \eqref{MM}.


\end{theorem}

\oomit{We now extend the PAC one-step safety certification of Lemma \ref{lemma:robust} to the multi-step setting, deriving the corresponding safety guarantees on the set $\mathbb{X}_s$.

\begin{theorem}[State-pointwise Finite-time PAC Safety Guarantee I]
\label{thm:vc}
    Let $\epsilon, \delta$, and $\gamma \in (0,1)$ be given constants. Let $\bm{a}^*(D^M)$ and $\xi^*(D^M)$ be the optimal solution to \eqref{eq:RBF_linear2} with $D^M=\{\bm{d}^{(j)}\}_{j=1}^M\stackrel{\text{i.i.d.}}{\sim}\textnormal{P}_{\bm{d}}$ \big(if $\xi^*(D^M)=0$, $\bm{a}^*(D^M)$ is obtained via solving \eqref{eq:RBF_linear21}\big), and  the certification procedure $\mathcal{A}(D^M)$ be defined as in \eqref{cp}. Then the following probabilistic guarantee holds:
\begin{equation*}
\begin{split}
   &\textnormal{P}_{\bm{d}}^M \left[
   \xi^*(D^M)=0 \Rightarrow
    \begin{split}
         &\Big(\textnormal{P}_{\pi}[\wedge_{i=1}^k \bm{\phi}_{\pi}^{\bm{x}}(i) \in \mathbb{X}_s]\geq (1-\epsilon(\bm{x}))^k\Big)
        \end{split}
        \right]\\
        &\geq 
          1-k\delta, \forall \bm{x}\in \mathbb{X},
\end{split}
\end{equation*}
where $\mathbb{X}_s=\{\bm{x}\in \mathbb{X}\mid h(\bm{a}^*(D^M),\bm{x}) > 0\} \subseteq \mathbb{X}$, $\epsilon(\bm{x})$ is defined in \eqref{epsilon}, and the sample number $M$ satisfies \eqref{MM}.
\end{theorem}
\begin{proof}
The system \eqref{eq:system} evolves as a Markov process.
By Lemma~\ref{lemma:robust}, for each fixed $\bm x\in\mathbb X$, with probability at least $1-\delta$ over the draw of $D^M$,
if $\xi^*(D^M)=0$ and $\bm x\in\mathbb X_s$, then
\[
\textnormal{P}_{\pi}\big[\bm\phi_{\pi}^{\bm x}(1)\in\mathbb X_s\big]\ge 1-\epsilon.
\]

\noindent\textbf{Step 1. (Union bound over $k$ steps).}
For the $k$ one-step conditions corresponding to
\(\bm\phi_{\pi}^{\bm x}(i)\in\mathbb X_s\) for $i=0,\ldots,k-1$,
define events
\[
\mathbb{E}_i := \big\{ D^M \mid
     \textnormal{P}_{\pi}\big[\bm\phi_{\pi}^{\bm x}(i+1)\in\mathbb X_s
          \mid \bm\phi_{\pi}^{\bm x}(i)\in\mathbb X_s \big]
     \ge 1-\epsilon
   \big\}.
\]
Each $\mathbb{E}_i$ occurs with probability at least $1-\delta$ by Lemma~\ref{lemma:robust} applied pointwise.
By the union bound,
\[
\textnormal{P}_{\bm{d}}^M\!\Big[\bigcap_{i=0}^{k-1} \mathbb{E}_i\Big]
   \ge 1 - \sum_{i=0}^{k-1}\textnormal{P}_{\bm{d}}^M[\mathbb{E}_i^c]
   \ge 1-k\delta.
\]

\noindent\textbf{Step 2. (Markov property).}
Condition on any realization $D^M$ inside $\bigcap_i \mathbb{E}_i$ and with $\xi^*(D^M)=0$.
For $\bm x\in\mathbb X_s$, the Markov property yields
\[
\textnormal{P}_{\pi}\!\Big[\bigwedge_{i=0}^{k-1}\bm\phi_{\pi}^{\bm x}(i+1)\in\mathbb X_s\Big]
= \prod_{i=0}^{k-1}
   \textnormal{P}_{\pi}\big[\bm\phi_{\pi}^{\bm x}(i+1)\in\mathbb X_s
                 \mid \bm\phi_{\pi}^{\bm x}(i)\in\mathbb X_s\big].
\]
Each factor is at least $1-\epsilon$ by the definition of $\mathbb{E}_i$,
hence the product is bounded below by $(1-\epsilon)^k$.

\noindent\textbf{Step 3. (Outer confidence).}
Since $\bigcap_i \mathbb{E}_i$ has probability at least $1-k\delta$ over $D^M$,
the overall confidence that the $k$-step safety probability exceeds $(1-\epsilon)^k$
whenever $\xi^*(D^M)=0$ is at least $1-k\delta$.
\end{proof}
}
\oomit{
\begin{proof}
The system \eqref{eq:system} evolves as a Markov process. By Lemma \ref{lemma:robust}, for each fixed $\bm{x} \in \mathbb{X}$, with probability at least $1-\delta$ over $D^M$, if $\xi^*(D^M)=0$ and $\bm{x} \in \mathbb{X}_s$, then:
\[
\text{P}_{\pi}[\bm{\phi}_{\pi}^{\bm{x}}(1) \in \mathbb{X}_s] \geq 1-\epsilon.
\]

Now consider the $k$-step safety probability. For $\bm{x} \in \mathbb{X}_s$ and $\xi^*(D^M)=0$, the Markov property implies:
\[
\begin{split}
&\text{P}_{\pi}[\wedge_{i=0}^{k-1} \bm{\phi}_{\pi}^{\bm{x}}(i+1) \in \mathbb{X}_s] \\
&= \prod_{i=0}^{k-1} \text{P}_{\pi}[\bm{\phi}_{\pi}^{\bm{x}}(i+1) \in \mathbb{X}_s \mid \bm{\phi}_{\pi}^{\bm{x}}(i) \in \mathbb{X}_s] \\
&\geq (1-\epsilon)^k,
\end{split}
\]
where the inequality follows because for any state $\bm{x}' \in \mathbb{X}_s$, the one-step safety probability from $\bm{x}'$ is at least $1-\epsilon$ (by the same reasoning as in Lemma \ref{lemma:robust} applied to $\bm{x}'$).

The confidence level $1-k\delta$ comes from applying the union bound over the $k$ time steps, where each step's guarantee holds with probability at least $1-\delta$ over $D^M$.
\end{proof}
}

\begin{remark}
\label{cover_X}
    In Theorem \ref{lemma:robust}, we provide safety certification only for $\bm{x} \in \mathbb{X}_s$. 
The definition \eqref{epsilon} on $\epsilon(\bm{x})$ ensures that the PAC safety guarantee holds over the entire domain $\mathbb{X}$, consistent with the formulation of Problem \ref{prob:1}, while trivially assigning zero probability of safety to states outside $\mathbb{X}_s$.
\end{remark}


The proof of Theorem~\ref{lemma:robust} is provided in Appendix~\ref{sec:proof_theorem_robust}. The above VC-dimension–based analysis applies to any class of barrier functions satisfying Assumption~\ref{ass:VC_finite}. However, the required sample complexity can be large, imposing substantial computational burden in practice. On the other hand, the one-step PAC guarantee in Theorem~\ref{lemma:robust} is state-pointwise: for any fixed state $\bm{x} \in \mathbb{X}$ it provides, with a certain confidence, a probability lower bound on the one-step safety probability originating from that state. This pointwise property does not directly extend to $k$-step safety for the actual (random) trajectory. 
{We will investigate this in the future work.} To alleviate these issues, we next focus on a special class of barrier functions that are linear in $\bm{a}$ and employ scenario approach \cite{calafiore2006scenario} to reduce the sample complexity and provide PAC $k$-step safety guarantees.

\subsubsection{Scenario Optimization Based Methods}
In this section, we employ scenario approach, originally proposed for solving robust optimization problems in \cite{calafiore2006scenario}, to reduce the sample complexity of the VC dimension-based method discussed above and provide PAC $k$-step safety guarantees, under Assumption \ref{as:h}.

\begin{assumption}
\label{as:h}
The parameterized function $h(\bm{a},\cdot): \mathbb{R}^n\rightarrow \mathbb{R}$ is linear over unknown parameters $\bm{a} \in \mathbb{R}^m$ with $h(\bm{0},\bm{x})\equiv 0$, and continuous with respect to $\bm{a}\in \mathbb{R}^{m}$ and $\bm{x}\in \mathbb{R}^n$. 
\end{assumption}

Under Assumption~\ref{as:h}, a conclusion stronger than Theorem~\ref{lemma:robust} can be established by solving \eqref{eq:RBF_linear21}. It provides a state-uniform PAC one-step safety guarantee while requiring fewer samples. 
\begin{lemma}[State-uniform PAC one-step safety guarantee]
\label{lemma:robust_s}
    Let $\epsilon, \delta$, and $\gamma \in (0,1)$ be given constants. Let $\bm{a}^*(D^M)$ be the optimal solution to \eqref{eq:RBF_linear21} with $D^M=\{\bm{d}^{(j)}\}_{j=1}^M\stackrel{\text{i.i.d.}}{\sim}\textnormal{P}_{\bm{d}}$, and the certification procedure $\mathcal{A}(D^M)$ be defined as in \eqref{cp}. Then, under Assumption \ref{as:h}, the following probabilistic guarantee holds:
    \begin{equation*}
\begin{split}
  & \textnormal{P}_{\bm{d}}^M \left[{\mathbb{X}_s\!\neq\! \emptyset}  \Rightarrow 
    \begin{split}
\Big(\forall \bm{x}\in \mathbb{X}, \textnormal{P}_{\bm{d}}\left[
\begin{split}
&h(\bm{a}^*(D^M),\bm{f}(\bm{x},\bm{d})) \\
& \geq\gamma h(\bm{a}^*(D^M),\bm{x})
\end{split}
\right]\geq 1-\epsilon \Big)
        \end{split}
        \right]\\   
          &\geq 
          1-\delta, 
\end{split}
\end{equation*}
implying 
\begin{equation*}
\begin{split}
   &\textnormal{P}_{\bm{d}}^M \left[
    \begin{split}
     {\mathbb{X}_s\!\neq\! \emptyset} \Rightarrow \Big(\forall \bm{x}\in \mathbb{X}_s, \textnormal{P}_{\bm{d}}[\bm{f}(\bm{x},\bm{d}) \in \mathbb{X}_s]\geq 1-\epsilon \Big)
        \end{split}
        \right]\geq 
          1-\delta,
\end{split}
\end{equation*}
where $\mathbb{X}_s=\{\bm{x}\in \mathbb{X}\mid h(\bm{a}^*(D^M),\bm{x}) > 0\}$ and the number $M$ satisfies \begin{equation}
    \label{MM_so}
    M \geq \frac{2}{\epsilon}(\ln{\frac{1}{\delta}}+m).
\end{equation}
\end{lemma}

\begin{theorem}[State-uniform PAC Finite-Time Safety Certification II]
\label{thm:so}
    Let $\epsilon, \delta$, and $\gamma \in (0,1)$ be given constants. Let $\bm{a}^*(D^M)$ be the optimal solution to \eqref{eq:RBF_linear21} with $D^M=\{\bm{d}^{(j)}\}_{j=1}^M\stackrel{\text{i.i.d.}}{\sim}\textnormal{P}_{\bm{d}}$, and  the certification procedure $\mathcal{A}(D^M)$ be defined as in \eqref{cp}. Then, under Assumption \ref{as:h}, the probabilistic guarantee holds:
 \begin{equation*}
\begin{split}
   &\textnormal{P}_{\bm{d}}^M \left[
    {\mathbb{X}_s\neq \emptyset} \Rightarrow
    \begin{split}
         \Big(\forall \bm{x}\in \mathbb{X}_s, \textnormal{P}_{\pi}[\wedge_{i=1}^k \bm{\phi}_{\pi}^{\bm{x}}(i) \in \mathbb{X}_s]\geq (1-\epsilon)^k\Big)
        \end{split}
        \right]\\
        &\geq 
          1-\delta,
\end{split}
\end{equation*}
where $\mathbb{X}_s=\{\bm{x}\in \mathbb{X}\mid h(\bm{a}^*(D^M),\bm{x}) > 0\}$ and the sample number $M$ satisfies \eqref{MM_so}.
\end{theorem}

The proofs of Lemma \ref{lemma:robust_s} and Theorem \ref{thm:so} are given in Appendix~\ref{sec:proof_lemma_robust_s} and \ref{sec:proof_thm_so}. Please refer to Remark \ref{cover_X} for the extension to all $\bm{x} \in \mathbb{X}$. 

Theorems~\ref{lemma:robust} and~\ref{thm:so} respectively provide 
state-pointwise PAC one-step and state-uniform PAC k-steps safety guarantees, 
when $\bm{x} \in \mathbb{X}_s$. 
It is noting that the state-uniform PAC guarantee in 
Lemma~\ref{lemma:robust_s} 
are stronger than the state-pointwise PAC guarantee in 
Theorem~\ref{lemma:robust}; 
the former implies the latter. 

The resulting set $\mathbb{X}_s$ is particularly valuable in safety-critical systems, 
as it delineates a highly trustworthy region in which the system can safely operate. 
Conceptually, $\mathbb{X}_s$ can be regarded as a set that plays a role analogous to an invariant set: 
if the system starts from an initial state within $\mathbb{X}_s$, 
it will remain inside this set over the next $k(\geq 1)$ steps 
with probability at least $(1-\epsilon)^k$ 
and confidence at least $1-\delta$. 
The safety certification can be applied recursively in an online manner,
providing continuous assurance of safe operation 
whenever the system remains within $\mathbb{X}_s$ over time. 

However, many systems may admit overly conservative $\mathbb{X}_s$ 
or even fail to admit a nonempty $\mathbb{X}_s$, 
particularly when the desired confidence and probability thresholds 
$\delta$ and $\epsilon$ are small. On the other hand, the robustness-based formulation \eqref{eq:RBF_linear21} can be overly conservative or even infeasible for many stochastic systems. This method enforces safety under all sampled disturbances, effectively performing a worst-case analysis over the empirical disturbance set. As a result, rare tail events are treated with the same importance as high-probability disturbances, which can lead to unnecessarily restrictive safety certificates. To overcome the above mentioned issues, we next develop an alternative approach for solving Problem \ref{prob:1} using stochastic barrier functions that explicitly incorporate disturbance probabilities.

\subsection{Certification via Stochastic Barrier Functions}
\label{sec:probabilistic}
In this section, we address the solution of Problem \ref{prob:1} with $k=1$ using stochastic barrier functions (SBFs). 

\begin{definition} 
    \label{def:sbf}
    Given a safe set $\mathbb{X}$, a function $h(\cdot):\mathbb{R}^n\rightarrow\mathbb{R}$ is a SBF for the system 
    \eqref{eq:system} if there exists a scalar $\lambda \in [0,1]$ that meets the following condition:
    \begin{subequations}
    \label{eq:sbf}
    \begin{empheq}[left=\empheqlbrace]{align}
        &h(\bm{x}) \geq 0, &\forall \bm{x}\in \mathbb{X}, \label{eq:sbf1} \\
        &h(\bm{x}) \geq 1, & \forall \bm{x}\in \overline{\widehat{\mathbb{X}}\setminus \mathbb{X}}, \label{eq:sbf2} \\
        &\textnormal{E}_{\bm{d}}[h(\bm{f}(\bm{x},\bm{d}))] - h(\bm{x}) \leq \lambda, &\forall \bm{x} \in \mathbb{X}. \label{eq:sbf3}
    \end{empheq}
    \end{subequations}
\end{definition}

\begin{proposition} [\cite{kushner1967}]
\label{finite_t}
    If there exist a SBF $h(\cdot): \mathbb{R}^n\rightarrow\mathbb{R}$ and  $\lambda\in[0,1]$, then \[\textnormal{P}_{\pi}\Big[\wedge_{i=1}^k \bm{\phi}_{\pi}^{\bm{x}}(i)\in \mathbb{X}\Big] \geq 1 - k\lambda - h(\bm{x}), \forall  \bm{x}\in\mathbb{X}. \] 
\end{proposition}

\begin{remark}
  In \cite{kushner1967}, the domain $\mathbb{R}^n$ is used in \eqref{eq:sbf1}–\eqref{eq:sbf3}, rather than the compact set $\widehat{\mathbb{X}}$ later introduced in \cite{xue2024finite}. Although the primary purpose of introducing $\widehat{\mathbb{X}}$ in \cite{xue2024finite} was to define a stopped process, restricting attention to a smaller compact set generally leads to more accurate results when performing approximations.
\end{remark}

In the following, we present our stochastic barrier function method for addressing Problem~\ref{prob:1} with $k=1$, leveraging the constraints \eqref{eq:sbf1}–\eqref{eq:sbf3} together with tools from statistical learning theory, in particular Rademacher complexity. 


Since the disturbance distribution $\textnormal{P}_{\bm{d}}$ is unknown, the expectation term $\textnormal{E}_{\bm{d}}[h(\bm{a}, \bm{f}(\bm{x},\bm{d}))]$ cannot be evaluated directly. To address this issue, we extract $M$ i.i.d disturbance samples of the form $D^M=\{\bm{d}^{(i)}\}_{i=1}^{M}$ and use these samples to empirically estimate the expectation, i.e., $\textnormal{E}_{\bm{d}}[h(\bm{a}, \bm{f}(\bm{x},\bm{d}))] \lesssim \frac{1}{M}\sum_{j=1}^{M}h(\bm{a}, \bm{f}(\bm{x},\bm{d}^{(j)})) + \tau$, where $\tau \in [0,1]$ is a deviation term accounting for the finite-sample error. This leads to the following  optimization problem \eqref{eq:SBF_linear11}:
\begin{equation}
    \label{eq:SBF_linear11}
  \begin{split}
    &\min_{\lambda \in \mathbb{R}, ~\bm{a}\in \mathbb{R}^m} {\textstyle\frac{1}{N_o}\sum_{i=1}^{N_o}\big(\lambda+h(\bm{a},\bm{x}'_i)\big)}\\
    \text{s.t.}&
    \begin{cases}
       h(\bm{a},\bm{x}) \geq 0, &\forall \bm{x}\in \mathbb{X},\\
       h(\bm{a},\bm{x})\geq 1, &\forall \bm{x}\in \overline{\widehat{\mathbb{X}}\setminus \mathbb{X}},\\
       \hat{h}(\bm{x},\bm{a},\lambda)\leq 0,  &\forall \bm{x} \in \mathbb{X}, \\
        \lambda \in [0,1]; ~ \|\bm{a}\|_2 \leq  U_a,
    \end{cases}
    \end{split}
\end{equation}
where  $\hat{h}(\bm{x},\bm{a},\lambda):=\frac{1}{M}\sum_{j=1}^{M}h(\bm{a}, \bm{f}(\bm{x},\bm{d}^{(j)})) -(h(\bm{a},\bm{x})  + \lambda - \tau)$,  $U_a \geq 1$ is a user-defined scalar value, and $\{\bm{x}'_i\}_{i=1}^{N_o}$ is a family of specified states distributed evenly over $\mathbb{X}$, which can be generated by random sampling from a uniform distribution over the state space $\mathbb{X}$. 

In addition, we can demonstrate that \eqref{eq:SBF_linear11} is feasible, i.e., there exists $(\lambda,\bm{a})$ satisfying \eqref{eq:SBF_linear11}. Any $h(\bm{a},\bm{x})\equiv U_a\in [0,U_a]$ for $\bm{x}\in \mathbb{X}$ with $\lambda= \tau$ will be a solution \eqref{eq:SBF_linear11}.

Before establishing the connection between the sample-based optimization \eqref{eq:SBF_linear11} and Problem~\ref{prob:1} with $k=1$, we first define a function class based on $h(\bm{a}, \bm{x})$ as used in \eqref{eq:SBF_linear11}. The Rademacher complexity of this function class provides a quantitative tool to link the empirical optimization with PAC finite-time safety guarantees.

\begin{assumption}
\label{norm_bound}
    For a fixed $\bm{x}$, define the function class: \[\mathbb{G}_{\bm{x}}=\Big\{h(\bm{a},\bm{f}(\bm{x},\cdot))=\bm{a}^T \cdot \bm{\phi}(\bm{x},\bm{d})\mid \|\bm{a}\|_2 \leq  U_a\Big\}.\] This class is linear in $\bm{a}$: for fixed $\bm{x}$ and $\bm{d}$, $h(\bm{a},\bm{f}(\bm{x},\bm{d}))$ is linear in $\bm{a}$.
Moreover, there exists a scalar value $R$ such that $\|\bm{\phi}(\bm{x},\bm{d})\|_2\leq R$ uniformly for all $\bm{x}\in \mathbb{X}$ and $\bm{d}\in \mathbb{D}$.
\end{assumption}

Under Assumption \ref{norm_bound}, a standard Rademacher complexity argument implies the following generalization guarantee. For each $\bm{x}\in\mathbb{X}$, with confidence at least $1-\delta$, the expected value
$\textnormal{E}_{\bm{d}}\big[h(\bm{a}^*(D^M), \bm{f}(\bm{x},\bm{d}))\big]$
is less than or equal to its empirical counterpart 
$\frac{1}{M}\sum_{j=1}^M h(\bm{a}^*(D^M), \bm{f}(\bm{x}, \bm{d}^{(j)}))$
plus $\tau$, where $\bm{a}^*(D^M)$ and $\lambda^*(D^M)$ are obtained by solving \eqref{eq:SBF_linear11}. Consequently, for each $\bm{x}\in\mathbb{X}$, with the same confidence, the probability that the system \eqref{eq:system} starting from $\bm{x}$ remains within the safe set $\mathbb{X}$ at the next step is at least 
$1 - \lambda^*(D^M) - h(\bm{a}^*(D^M), \bm{x})$. The above statements hold provided that the sample size satisfies
 \begin{equation}
    \label{MM_RA}
        M \geq  
    \frac{U_a^2}{\tau^2}\left(
         2R +\sqrt{2\ln \frac{1}{\delta}}
    \right)^{\!2}.
    \end{equation}

\begin{theorem}[State-pointwise PAC one-step safety guarantee II]
    \label{lemma:1}
    Let $\delta \in (0,1)$, $U_a \geq 1$, and $\tau \in [0, 1]$ be given constants.  
    Let $\big(\bm{a}^*(D^M),\lambda^*(D^M)\big)$ be the (locally/globally) optimal solution to \eqref{eq:SBF_linear11} with 
    $D^M=\{\bm{d}^{(j)}\}_{j=1}^M\stackrel{\text{i.i.d.}}{\sim}\textnormal{P}_{\bm{d}}$, and let
\begin{equation*}
\label{cp1}
\mathcal{A}(D^M) := 1 (\text{certificate always accepted})
\end{equation*}
be the certification procedure.
    Then, under Assumption \ref{norm_bound}, the following probabilistic guarantee holds: 
    \begin{equation*}
    \begin{split}
    \textnormal{P}_{\bm{d}}^M \left[ 
    \begin{split}
       & \textnormal{E}_{\bm{d}}[h(\bm{a}^*(D^M), \bm{f}(\bm{x},\bm{d}))]\leq \\
        &\tfrac{1}{M}\sum_{j=1}^M h(\bm{a}^*(D^M), \bm{f}(\bm{x}, \bm{d}^{(j)})) +\tau
        \end{split}
    \right]\geq 1-\delta, \forall \bm{x}\in \mathbb{X}, 
    \end{split}
    \end{equation*}
    provided that the sample size $M$ satisfies \eqref{MM_RA}.  Further, 
    \[
    \begin{split}
    &\textnormal{P}_{\bm{d}}^M \left[
    \begin{split}
       \textnormal{P}_{\bm{d}} \big[\bm{f}(\bm{x},\bm{d})\in \mathbb{X}\big]         \geq  1-\lambda^*(D^M)-h(\bm{a}^*(D^M), \bm{x})
        \end{split}
    \right]\\
    &\geq 1-\delta, \forall \bm{x}\in \mathbb{X}
    \end{split}
    \]
    holds.
\end{theorem}

\oomit{Based on Proposition \ref{finite_t}, we now extend the PAC one-step safety certification of Lemma \ref{lemma:1} to the multi-step setting.
\begin{theorem}[State-pointwise PAC Finite-Time Safety Certification III]
    \label{theo:1}
    Let $\delta \in (0,1)$, $\epsilon \in (0,1)$, $U_a \geq 1$, $\tau \in [0, 1]$, and $\big(\bm{a}^*(D^M), \lambda^*(D^M)\big)$ be the feasible solution to \eqref{eq:SBF_linear11} with 
    $D^M=\{\bm{d}^{(j)}\}_{j=1}^M\stackrel{\text{i.i.d.}}{\sim}\textnormal{P}_{\bm{d}}$, and the certification procedure $\mathcal{A}(D^M)$ be defined as 
    \begin{equation*}
\label{cp00}
\mathcal{A}(D^M) := 
\begin{cases}
1, & \text{if } \lambda^*(D^M)+h(\bm{a}^*(D^M),\bm{x})\leq \epsilon, \forall \bm{x}\in \mathbb{X},\\
0, & \text{otherwise.}
\end{cases}
\end{equation*}
    Then, under Assumption \ref{norm_bound}, the following probabilistic guarantee holds:  
    \[\begin{split}
    \textnormal{P}_{\bm{d}}^M \bigg[
     \begin{split}
        &\textnormal{P}_{\pi} [
       \wedge_{i=0}^{k-1} \bm{\phi}_{\pi}^{\bm{x}}(i+1) \in \mathbb{X}\big] \\
        &\geq  (1-\epsilon)^k
        \end{split}
    \bigg]\geq 1-k\delta, \forall \bm{x}\in \mathbb{X}, 
    \end{split}
    \]
    holds, 
    provided that the sample size $M$ satisfies \eqref{MM_RA}.
\end{theorem}
\begin{proof}
For each $\bm{x}_i\in \mathbb{X}$, $i\in\{0,\dots,k-1\}$, define the event
\[
\mathbb{E}_i := \Big\{ D^M \mid \textnormal{P}_{\bm d}\big[\bm f(\bm x_i,\bm d)\in\mathbb X\big]
    \ge 1-\lambda^*(D^M)-h(\bm a^*(D^M),\bm x_i) \Big\}.
\]
By the pointwise lemma we have $\textnormal{P}_{\bm{d}}^M [\mathbb{E}_i]\ge 1-\delta$ for every fixed $i$.
Applying the union bound gives
\[
\textnormal{P}_{\bm{d}}^M \Big[\bigcap_{i=0}^{k-1} \mathbb{E}_i\Big] \ge 1-\sum_{i=0}^{k-1}\textnormal{P}_{\bm{d}}^M [\mathbb{E}_i^c]
\ge 1-k\delta.
\]

Restrict attention to the event $\bigcap_{i=0}^{k-1} \mathbb{E}_i$, which occurs with probability at least $1-k\delta$.
On this event, for every $i$,
\[
\textnormal{P}_{\bm d}\big[\bm f(\bm x_i,\bm d)\in\mathbb X\big]
    \ge 1-\lambda^*(D^M)-h(\bm a^*(D^M),\bm x_i).
\]
If additionally $\mathcal{A}(D^M)=1$, then $\lambda^*(D^M)+h(\bm a^*(D^M),\bm x_i)\le\epsilon$ for all $i$, hence
\[
\textnormal{P}_{\bm{d}}^M \big[\bm f(\bm x_i,\bm d)\in\mathbb X\big] \ge 1-\epsilon
\qquad\text{for every }i=0,\dots,k-1.
\]

Using the Markov property of the dynamics, the $k$-step survival probability from $\bm x_0$ under policy $\pi$ factors as
\[
\textnormal{P}_{\pi}\!\Big[\bigwedge_{i=0}^{k-1}\bm\phi_{\pi}^{\bm x}(i+1)\in\mathbb{X}\Big]
= \prod_{i=0}^{k-1}
\textnormal{P}_{\pi}\big[\bm\phi_{\pi}^{\bm x}(i+1)\in\mathbb{X}
                 \mid \bm\phi_{\pi}^{\bm x}(i)\in\mathbb{X}\big].
\]
Each factor is at least $1-\epsilon$ on the event above, so the product is at least $(1-\epsilon)^k$.
Therefore, on the event $\bigcap_{i} \mathbb{E}_i$ (which has probability $\ge 1-k\delta$), whenever $\mathcal{A}(D^M)=1$ the $k$-step survival probability is $\ge(1-\epsilon)^k$. This proves the claimed statement.
\end{proof}
}
\oomit{
\begin{proof}
From Lemma~\ref{lemma:1}, for any $\bm{x}_0 \in \mathbb{X}$, with confidence at least $1-\delta$ over the sampling of $D^M$, the pair $(\bm{a}^(D^M), \lambda^(D^M))$ satisfies
 \begin{equation*}
    \begin{cases}
        h(\bm{a}^*(D^M),\bm{x}) \geq 0, &\forall \bm{x}\in \mathbb{X}, \\
        h(\bm{a}^*(D^M),\bm{x}) \geq 1, & \forall \bm{x} \in \widehat{\mathbb{X}}\setminus\mathbb{X},  \\
\textnormal{E}_{\bm{d}}[h(\bm{a}^*(D^M),\bm{f}(\bm{x}_0,\bm{d}))] \\
\quad\quad\quad\quad\quad\quad- h(\bm{a}^*(D^M),\bm{x}_0) \leq \lambda^*(D^M).
    \end{cases}
    \end{equation*}

Further, according to Proposition \ref{finite_t}, we can obtain that, with confidence at least $1-\delta$ over the sampling process on $D^M$, if $\mathcal{A}(D^M)$ is accepted,
\[\textnormal{P}_{\pi} [
       \bm{\phi}_{\pi}^{\bm{x}_0}(1) \in \mathbb{X}\big] \geq  1-\lambda^*(D^M)-h(\bm{a}^*(D^M), \bm{x}_0) \geq 1-\epsilon.\]

Using the elementary inequality $\textnormal{P}[A \cap B] \ge \textnormal{P}[A] + \textnormal{P}[B] - 1$, we have that, with confidence at least $1 - 2\delta$ over the sampling of $D^M$, if $\mathcal{A}(D^M)$ is accepted,
\[\textnormal{P}_{\pi} [
       \bm{\phi}_{\pi}^{\bm{x}_0}(1) \in \mathbb{X}\big] \geq 1-\epsilon\]
       and 
       \[\textnormal{P}_{\pi} [
       \bm{\phi}_{\pi}^{\bm{x}_0}(2) \in \mathbb{X}\mid \bm{\phi}_{\pi}^{\bm{x}_0}(1) \in \mathbb{X}\big] \geq  1-\epsilon.\]
Hence, with confidence at least $1 - 2\delta$ over the sampling process, if $\mathcal{A}(D^M)$ is accepted,
\[\textnormal{P}_{\pi} [
    \wedge_{i=1}^2\bm{\phi}_{\pi}^{\bm{x}_0}(i) \in \mathbb{X}\big] \geq (1-\epsilon)^2.\]

    Applying the same reasoning inductively yields that, with confidence at least $1 - k\delta$,
if $\mathcal{A}(D^M)$ is accepted,
\[\textnormal{P}_{\pi} [
    \wedge_{i=1}^k\bm{\phi}_{\pi}^{\bm{x}_0}(i) \in \mathbb{X}\big] \geq (1-\epsilon)^k.\]
    
Since this holds for any $\bm{x}_0\in\mathbb{X}$, the result follows.
\end{proof}
}

\color{black}
\oomit{\begin{remark}
It follows from Theorem~\ref{theo:1} that when $\lambda^*(D^M)=0$, we have $k\lambda^*(D^M)=0$.
Hence, under Assumption \ref{norm_bound}, the finite-time safety certification result can be extended to the infinite-time case according to Proposition \ref{finite_t}, where the following probabilistic guarantee holds:
 \[\begin{split}
    \textnormal{P}_{\bm{d}}^M \bigg[\lambda^*(D^M)=0 \Rightarrow 
    \begin{split}
    \Big(\bm{x}\in \mathbb{X}, &\textnormal{P}_{\pi} \big[\forall i\in \mathbb{N}. \bm{\phi}_{\pi}^{\bm{x}}(i) \in \mathbb{X}\big] \\
        &\geq  1-h(\bm{a}^*(D^M), \bm{x})\Big)
    \end{split}
    \bigg]\geq 1-\delta
    \end{split}
    \]
    holds,
    provided that the sample size satisfies
    \[
    M \geq  
    \frac{U_a^2}{\tau^2}\left(
         2R +\sqrt{2\ln \frac{1}{\delta}}
    \right)^{\!2}.
    \]
Here, the certification procedure
\begin{equation}
\mathcal{A}(D^M) :=
\begin{cases}
1 & \text{if } \lambda^*(D^M)=0,\\
0, & \text{otherwise.}
\end{cases}
\end{equation}

\end{remark}
}

The proof of Theorem~\ref{lemma:1} is presented in Appendix~\ref{sec:proof_lemma_1}.


The SBF method provides a state-dependent one-step safety lower bound $1 - \lambda^*(D^M) - h(\bm{a}^*(D^M), \bm{x})$. 
As with Theorem~\ref{lemma:robust}, the conclusion of Theorem~\ref{lemma:1} cannot be directly extended to multi-step safety guarantees. {We will investigate this in the future work based on SBFs proposed in \cite{xue2024finite}.} Nevertheless, the one-step safety certification can be applied recursively online, providing continuous assurance of safe operation at each time step, as long as the system remains within $\mathbb{X}$. Unlike the RBF method, the SBF method does not generally yield a probabilistic invariant-type set 
\(\mathbb{X}_s\).
The main reason is that the guarantee provided by the SBF method depends strongly on the value of the state-dependent barrier function $h$.
Note that the one-step lower bound represents a probability and should therefore be non-negative; 
however, the expression 
\(1 - \lambda^*(D^M) - h(\bm{a}^*(D^M), \bm{x})\) 
can become negative, in which case the bound is trivially true but uninformative.
Hence, the bound is practically useful only when it is positive. On the other hand, the SBF method relaxes the barrier condition: 
it does not require the condition to hold for every disturbance sample, 
but only in expectation (approximated using finitely many samples). 
This relaxation reduces the number of constraints in the optimization problem from \(M\) constraints in the RBF formulation to a single expectation constraint and can thus significantly reduce computational overhead, 
especially when a large number of samples is used.

The above Rademacher-complexity-based analysis applies to any class of barrier functions that possesses uniformly finite Rademacher complexity over all $\bm{x}\in \mathbb{X}$, including, in particular, classes that are uniformly Lipschitz continuous or realized by neural networks with bounded weights and Lipschitz activations. The optimization problem \eqref{eq:SBF_linear11}, formulated using disturbance samples, is not limited to convex forms and may, in general, be non-convex—posing intrinsic challenges for practical solution. The development of new algorithms to tackle such non-convex problems lies beyond the scope of this work. In the subsequent numerical study, we therefore concentrate on convex formulations, particularly semi-definite programs obtained via sum-of-squares decompositions of multivariate polynomials, owing to the availability of mature algorithms and software tools that can efficiently solve such problems.

\begin{remark}[Relationship with existing works]
\label{relation}

Based on the combination of scenario approach with piecewise affine stochastic barrier certificates 
(i.e., \(h(\bm{a}_i,b_i,\bm{x})=\bm{a}_i^{\top}\bm{x}+b_i, i=1,\ldots,l\)), 
\cite{mathiesen2023inner} considered piecewise affine systems, i.e.,
\[
    \bm{x}(j+1) = A_i \bm{x}(j) + \bm{d}(j), 
    \quad \text{for } \bm{x}(j) \in \mathbb{P}_i, \; i = 1, \ldots, l,
\]
where each region $\mathbb{P}_i$ is a polytope. 
\cite{mathiesen2024data} further studied piecewise affine approximations of nonlinear dynamics of the form 
\(\bm{x}(j+1) = \bm{f}(\bm{x}(j)) + \bm{d}(j)\).  
As discussed in the related work, they relaxed the conditions for stochastic barrier certificates 
into a chance-constrained formulation and subsequently constructed a scenario optimization problem 
that accounts for the worst-case disturbance samples---similar in spirit to our robust barrier 
certificate based approach proposed in this paper---for solving the chance-constrained  problem.

When considering piecewise affine systems or systems that can be discretized into piecewise affine forms, 
the use of piecewise stochastic barrier certificates allows us to obtain stronger results than those in 
Theorem~\ref{lemma:1}. 
In this context, the condition 
\[
\textnormal{E}_{\bm{d}}\!\left[h(\bm{a}_i,b_i,A_i\bm{x}+\bm{d})\right] 
\leq h(\bm{a}_i,b_i,\bm{x}) + \lambda, \quad \forall \bm{x}\in \mathbb{P}_i,
\]
can be written as
 $\bm{a}_i^{\top}A_i\bm{x} + \bm{a}_i^{\top}\textnormal{E}_{\bm{d}}[\bm{d}] + b_i
\leq h(\bm{a}_i,b_i,\bm{x}) + \lambda, \forall \bm{x}\in \mathbb{P}_i$, which is equivalent to 
\[
\bm{a}_i^{\top}\textnormal{E}_{\bm{d}}[\bm{d}] 
\leq h(\bm{a}_i,b_i,\bm{x}) + \lambda-\bm{a}_i^{\top}A_i\bm{x}-\bm{b}_i, \quad \forall \bm{x}\in \mathbb{P}_i,
\]
where $\bm{a}_i \in [-U_a,U_a]^n$ and $b_i\in [-U_a,U_a]$ for $i=1,\ldots,l$, and $\mathbb{X}=\cup_{i=1}^l \mathbb{P}_i$ with $\mathbb{P}_{i}^{\circ}\cap \mathbb{P}_{j}^{\circ}=\emptyset$ for $i\neq j$.

As a result, the discrepancy between the expected value, i.e., $\textnormal{E}_{\bm{d}}\!\left[h(\bm{a}_i,b_i,A_i\bm{x}+\bm{d})\right]$, and its empirical counterpart, i.e., $\frac{1}{M} \sum_{j=1}^M h(\bm{a}_i,b_i,A_i\bm{x}+\bm{d}^{(j)})$, only involves the term $\bm{a}_i^{\top}\textnormal{E}_{\bm{d}}[\bm{d}]-\bm{a}_i^{\top} \frac{1}{M}\sum_{j=1}^M \bm{d}^{(j)}$, which is independent of \(\bm{x}\). 
Thus, the conclusion in Lemma~\ref{lemma:1} can be strengthened from a 
\emph{state-pointwise} PAC one-step safety guarantee to a 
\emph{state-uniform} PAC one-step safety guarantee, as established in Lemma~\ref{lemma:robust_s}, i.e.,
\begin{equation*}
\begin{split}
&\textnormal{P}_{\bm{d}}^M \left[
\begin{split}
 &\forall i\in \{1,\ldots,l\}, \;\\
&\bm{a}^{*\top}_i(D^M)^{\top}\textnormal{E}_{\bm{d}}[\bm{d}] \leq 
    \bm{a}_i^{*\top}(D^M) \frac{1}{M}\sum_{j=1}^M \bm{d}^{(j)}+ \tau
\end{split}
\right]\geq 1-\delta,
\end{split}
\end{equation*}
 and thus
\begin{equation*}
\begin{split}
&\textnormal{P}_{\bm{d}}^M \!\left[
\begin{split}
  &\forall i\in \{1,\ldots,l\}, \forall \bm{x}\in \mathbb{P}_i, \bm{a}^{*\top}_i(D^M)^{\top}\textnormal{E}_{\bm{d}}[\bm{d}] \leq \lambda^*(D^M)+\\
  &h(\bm{a}_i^*(D^M),b_i^*(D^M),\bm{x}) -\bm{a}_i^{*\top}(D^M)A_i\bm{x}-b_i^*(D^M)
\end{split}
\right] \\
&\geq 1-\delta,
\end{split}
\end{equation*}
implying
\[
\begin{split}
&\textnormal{P}_{\bm{d}}^M \left[
\begin{split}
    &\forall i\in \{1,\ldots,l\}, \forall \bm{x}\in \mathbb{P}_i, \\
    &\textnormal{P}_{\bm{d}}\!\left[\bm{f}(\bm{x},\bm{d})\in \mathbb{X}\right]
    \geq 1-\lambda^*(D^M) - h(\bm{a}_i^*(D^M), b_i^*(D^M),\bm{x})
    \end{split}
\right] \\
&
\geq 1-\delta,
\end{split}
\]
holds under suitable conditions on the sample size \(M\). 
Correspondingly, the conclusion in Theorem~\ref{lemma:1} can be enhanced to the \emph{state-uniform} multi-step safety guarantees:
\[
\begin{split}
\textnormal{P}_{\bm{d}}^M \left[
\begin{split}
    &\forall i\in \{1,\ldots,l\}, \forall \bm{x}\in \mathbb{P}_i, \textnormal{P}_{\pi} \!\left[
        \bigwedge_{i=0}^{k-1} 
        \bm{\phi}_{\pi}^{\bm{x}}(i+1) \in \mathbb{X}
    \right] \\
    &
    \geq 1-k\lambda^*(D^M)-h(\bm{a}_i^*(D^M),b_i^*(D^M),\bm{x})
    \end{split}
\right] \geq 1 - \delta.
\end{split}
\]
where $\mathcal{A}(D^M) := 1 (\text{certificate always accepted})$ and $k>1$.

In addition, compared with the chance-constrained optimization method, 
our approach does not require the barrier condition to hold for every sampled disturbance realization. 
Instead, it only requires the condition to hold in expectation (with a finite-sample approximation), 
thus significantly reducing the computational burden, especially when a large number of samples is used. A detailed numerical comparison with the methods in \cite{mathiesen2023inner,mathiesen2024data} is deferred, 
as their implementation is not publicly available.
Furthermore, we plan to combine VC dimension /Rademacher complexity with a finite $\varepsilon$-covering (i.e., covering number) of the safe space $\mathbb{X}$, along with Lipschitz continuity, to extend the current state-pointwise PAC safety guarantees in Theorem~\ref{lemma:robust} and Theorem~\ref{lemma:1} to state-uniform PAC safety guarantees for nonlinear systems beyond the piecewise-affine setting in future work.

\end{remark}

\oomit{\subsubsection*{Bernstein barrier functions}
Here, we introduce a representative class of functions—polynomials expressed in the Bernstein basis—which are particularly useful in reducing sample complexity.

Assume $\widehat{\mathbb{X}}=(\underline{\widehat{\mathbb{X}}}[1],\overline{\widehat{\mathbb{X}}}[1])\times \cdots \times (\underline{\widehat{\mathbb{X}}}[n],\overline{\widehat{\mathbb{X}}}[n])$. We parameterize $h(\bm{a},\bm{x})$ in the following form, which satisfies Assumption \ref{as:h}:
\begin{equation}
\label{eq:h_p}
h(\bm{a},\bm{x})=\sum_{i_1=0}^{\kappa}\ldots\sum_{i_n=0}^{\kappa} a_{i_1,\ldots,i_n} \prod_{j=1}^n \dbinom{{\kappa}}{i_j} \psi_j^{i_j}(1-\psi_j)^{{\kappa}-i_j}
\end{equation}
with $a_{i_1,\ldots,i_n} \in [0, U_a]$ and $\psi_j = \frac{\bm{x}[j]-\underline{\widehat{\mathbb{X}}}[j]}{\overline{\widehat{\mathbb{X}}}[j]-\underline{\widehat{\mathbb{X}}}[j]}$. 
Parameterization in the form of \eqref{eq:h_p} allows constraint \eqref{eq:sbf1} to hold naturally. Moreover, parameterization in the form of \eqref{eq:h_p} ensures that $h(\bm{a},\bm{x}) \leq U_a$ holds for all $\bm{x} \in \widehat{\mathbb{X}}$ \cite{farouki2012bernstein}, which is important for the validity of Theorem \ref{theo:2}.

Assume the optimal solution $(\bm{a}^*(D^M),\lambda^*(D^M))$ to \eqref{eq:SBF_linear11} has been computed. Then, by Rademacher complexity for uniform coverage bounds, we can conclude that for each $\bm{x}$, with confidence at least $1-\delta$, the probability that the system \eqref{eq:system} starting from $\bm{x}$, remains within the safe set $\mathbb{X}$ at the next step is at least $1-\lambda^*(D^M)-h(\bm{a}^*(D^M),\bm{x})$, provided that $M \geq \frac{U_a^2}{\tau^2} \left( 2\sqrt{m} + \sqrt{2 \ln \frac{1}{\delta}} \right)^2$ holds. 
We have the following conclusion. 

\begin{theorem}
    \label{theo:2}
    Let $\epsilon, \delta \in (0,1)$, $U_a \geq 1$, and $\tau \in (0, 1)$ be given constants.  Let $\bm{a}^*(D^M)$ and $\lambda^*(D^M)$ be the optimal solution to \eqref{eq:SBF_linear11} with $D^M=\{\bm{d}^{(j)}\}_{j=1}^M\stackrel{\text{i.i.d.}}{\sim}\textnormal{P}_{\bm{d}}$. Then the following probabilistic guarantee holds: 
    \begin{equation}
    \begin{split}
    \textnormal{P}_{\bm{d}}^M \left[
    \begin{split}
    &\textnormal{E}_{\bm{d}}[h(\bm{a}^*(D^M), \bm{f}(\bm{x},\bm{d}))]\\
\leq &h(\bm{a}^*(D^M), \bm{x}) + \lambda^*(D^M)
\end{split}
\right]\geq 1-\delta, \forall \bm{x}\in \mathbb{X},
\end{split}
\end{equation}
provided that the sample size satisfies
    \begin{equation*}
        \begin{split}
            M \geq \frac{U_a^2}{\tau^2} \left( 2\sqrt{m} + \sqrt{2 \ln \frac{1}{\delta}} \right)^2,
        \end{split}
    \end{equation*}
    which implies 
\[\textnormal{P}_{\bm{d}}^M \left[
\begin{split}
    \textnormal{P}_{\bm{d}} &[\bm{f}(\bm{x},\bm{d})\in \mathbb{X}] \\
    &\geq 1-\lambda^*(D^M)-h(\bm{a}^*(D^M), \bm{x}) 
    \end{split}
    \right]\geq 1-\delta, \forall \bm{x}\in \mathbb{X}.
    \]
\end{theorem}
\begin{proof}
For a fixed $\bm{x}$, define the function class: \[\mathbb{G}_{\bm{x}}=\{h(\bm{a},\bm{f}(\bm{x},\cdot))=\bm{a}^T \cdot \bm{\phi}(\bm{x},\bm{d})\mid \bm{a}\in [0,U_a]^m\}.\]

This class is:
\begin{enumerate}
    \item Linear in $\bm{a}$: $h(\bm{a},\bm{f}(\bm{x},\bm{d}))$ is linear in $\bm{a}$ for fixed $\bm{x}$ and $\bm{d}$;
    \item Bounded: $h(\bm{a},\bm{f}(\bm{x},\bm{d}))\in [0,U_a]$ and  $\|\bm{\phi}(\bm{x},\bm{d})\|_2\leq 1$ for $\bm{a}\in [0,U_a]^m$, $\bm{d}\in \mathbb{D}$ and $\bm{x}\in \widehat{\mathbb{X}}$(This is guaranteed, since $h(\bm{a},\bm{x})$ is a polynomial with a Bernstein basis).
\end{enumerate}

The empirical Rademacher complexity of $\mathbb{G}_{\bm{x}}$ for a sample $D^M$ is :
\[\Re_M(\mathbb{G}_{\bm{x}})=\frac{1}{M}\textnormal{E}_{\sigma}\Big[\sup_{\bm{a}\in [-U_{\bm{a}},U_{\bm{a}}]^m} \big|\sum_{j=1}^M \sigma_jh(\bm{a},\bm{f}(\bm{x},\bm{d}^{(j)}))\big|\Big],\]
where $\sigma_j$ are independent Rademacher random variables ($P(\sigma_i=+1)=P(\sigma_j=-1)=\frac{1}{2}$). 

Using standard results for bounded linear function classes, we have $\Re_M(\mathbb{G}_{\bm{x}})\leq \frac{\sqrt{m}U_a}{\sqrt{M}}$, $\forall \bm{x}\in \mathbb{X}$. It is obtained as follows: 
 By linearity of $h$, we have
\[
\sum_{j=1}^M \sigma_j h(\bm{a}, \bm{f}(\bm{x}, \bm{d}^{(j)})) = \bm{a}^\top \sum_{j=1}^M \sigma_j \bm{\phi}(\bm{x}, \bm{d}^{(j)}).
\]
Using Cauchy--Schwarz, we obtain
\[
\begin{split}
&\sup_{\bm{a} \in [0,U_a]^m} \Big| \bm{a}^\top \sum_{j=1}^M \sigma_j \bm{\phi}(\bm{x}, \bm{d}^{(j)}) \Big| \\
&\leq \sup_{\bm{a} \in [0,U_a]^m} \|\bm{a}\|_2 \Big\| \sum_{j=1}^M \sigma_j \bm{\phi}(\bm{x}, \bm{d}^{(j)}) \Big\|_2 
\\
&\leq U_a \sqrt{m} \Big\| \sum_{j=1}^M \sigma_j \bm{\phi}(\bm{x}, \bm{d}^{(j)}) \Big\|_2.
\end{split}
\]
By Jensen's inequality and the properties of Rademacher variables, we have
\[
\begin{split}
&\textnormal{E}_\sigma \Big\| \sum_{j=1}^M \sigma_j \bm{\phi}(\bm{x}, \bm{d}^{(j)})  \Big\|_2 \leq \sqrt{\textnormal{E}_\sigma \Big\| \sum_{j=1}^M \sigma_j \bm{\phi}(\bm{x}, \bm{d}^{(j)})  \Big\|_2^2} \\
&= \sqrt{\sum_{j=1}^M \|\bm{\phi}(\bm{x}, \bm{d}^{(j)}) \|_2^2} \leq \sqrt{M}.
\end{split}
\]
Hence, we obtain
\[
\Re_M(\mathbb{G}_{\bm{x}}) \leq \frac{1}{M} U_a \sqrt{m} \sqrt{M} = \frac{U_a \sqrt{m}}{\sqrt{M}}, \forall \bm{x}\in \mathbb{X}.
\]

By the one-sided Rademacher concentration inequality \cite{mitzenmacher2017probability}, we can conclude, for any $t>0$, 
\begin{equation*}
    \begin{split}
        &\textnormal{P}_{\bm{d}}^M \left[ 
        \begin{split}
        &\sup_{\bm{a}\in [0,U_a]^m}\big( \textnormal{E}_{\bm{d}}[h(\bm{a},\bm{f}(\bm{x},\bm{d}))] - \frac{1}{M} \sum_{j=1}^{M} h(\bm{a}, \bm{f}(\bm{x}, \bm{d}^{(j)}))\big) \\
        &>2\Re_M(\mathbb{G}_{\bm{x}})+t
        \end{split}
        \right]\\
        &
        \leq   e^{-\frac{M t^2}{2U_a^2}}, \forall \bm{x}\in \mathbb{X}.
    \end{split}
\end{equation*}

We choose $t$ such that this probability is bounded by $\delta$:
\[
\exp\left( -\frac{M t^2}{2U_a^2} \right) \leq \delta \Rightarrow t \geq U_a \sqrt{\frac{2}{M} \ln \frac{1}{\delta}}.
\]
We now require that the total deviation is at most $\tau$:
\[
2\Re_M(\mathbb{G}_{\bm{x}}) + t \leq \tau.
\]
Using the bound $\Re_M(\mathbb{G}_{\bm{x}}) \leq \frac{U_a \sqrt{m}}{\sqrt{M}}$ and the minimal choice for $t$, we obtain a sufficient condition:
\[
2 \cdot \frac{U_a \sqrt{m}}{\sqrt{M}} + U_a \sqrt{\frac{2}{M} \ln \frac{1}{\delta}} \leq \tau.
\]
Multiplying both sides by $\sqrt{M}/U_a$, we have
\[
2 \sqrt{m} + \sqrt{2 \ln \frac{1}{\delta}} \leq \frac{\tau \sqrt{M}}{U_a}.
\]
Solving for $M$ gives the condition, we have
\[
M \geq \frac{U_a^2}{\tau^2} \left( 2\sqrt{m} + \sqrt{2 \ln \frac{1}{\delta}} \right)^2.
\]

Following the proof of Theorem \ref{theo:1}, we have the conclusion.
\end{proof}
}

\section{Examples}
\label{sec:exam}

In this section, we demonstrate the effectiveness of the proposed approaches through several numerical examples. All examples involve polynomial systems, where each component of 
$\bm{f}(\bm{x},\bm{d})$ is polynomial in $\bm{x}$ and the safe set $\mathbb{X}$ is semi-algebraic. This choice allows us to leverage advanced computational tools such as sum-of-squares (SOS) programming, implemented via YALMIP \cite{lofberg2004yalmip} and Mosek 11.0.29 \cite{aps2019mosek}, to solve optimization problems \eqref{eq:RBF_linear21} and \eqref{eq:SBF_linear11} for polynomial barrier functions. For details on SOS programming, please refer to Appendix \ref{sec:app_sos}. All computations are run on a machine equipped with an i7-13700 2.1GHz CPU with 32GB RAM.


Although this work focuses on stochastic systems with unknown disturbance distributions, two non-data-driven baseline methods are included for comparison. 
When the disturbance set $\mathbb{D}$ is known and semi-algebraic, the constraint \eqref{eq:rbf} can be reformulated as an SOS programming, yielding a robust invariant set that guarantees forward invariance under all disturbances in $\mathbb{D}$. We refer to this baseline as the \textit{Robust-Set} method and compare it with the proposed Data-Driven RBFs method. Similarly, if the disturbance distribution is assumed to be known, the constraint \eqref{eq:sbf} can be recast as an SOS programming. We refer to this baseline as the \textit{Distribution-Known} method and compare it with the proposed Data-Driven SBFs method. To compare the Robust-Set method with the Data-Driven RBFs method, we estimate the relative volume of the resulting set $\mathbb{X}_s$ using Monte Carlo sampling. Specifically, $10^7$ samples are drawn uniformly from $\mathbb{X}$, and 
$V_{\mathbb{X}_s} = \frac{N_{\mathbb{X}_s}}{10^7}$
is computed, where $N_{\mathbb{X}_s}$ is the number of samples contained in $\mathbb{X}_s$. A larger $V_{\mathbb{X}_s}$ indicates a larger set $\mathbb{X}_s$. For comparison between the Distribution-Known and the proposed Data-Driven SBFs method, we use the optimal value $J^*=\frac{1}{N_o}\sum_{i=1}^{N_o}\big(\lambda^*+h(\bm{a}^*,\bm{x}'_i)\big)$ in \eqref{eq:SBF_linear11}, where a smaller $J^*$ is preferred. We also report the total computation time $T_{\mathrm{total}} = T_s + T_c$, where $T_s$ and $T_c$ denote the sampling/preprocessing time and the optimization time, respectively. All times are measured in seconds.

For the Data-Driven RBFs and Robust-Set methods, we fix $\gamma=0.99$ and $U_{\bm{a}}=10^4$. In all experiments, we set $\delta = 0.001$ for both Data-Driven RBFs and Data-Driven SBFs methods, corresponding to a $99.9\%$ confidence level according to Theorem~\ref{lemma:robust}, \ref{thm:so}, and \ref{lemma:1}, ensuring a high level of reliability in the results. 
Since the sample-size requirement in Theorem~\ref{thm:so} is always less restrictive than that in Theorem~\ref{lemma:robust} under the same settings, the number of samples in all experiments is chosen as the smallest value that satisfies Theorem~\ref{thm:so}.

\begin{example}
    \label{ex:bc} Consider the following system adapted from \cite{prajna2004safety},
    \begin{equation*}
        \begin{cases}
        x(t+1)=x(t) + 0.1(y(t)+x(t)d(t)),\\
        y(t+1)=y(t) + 0.1(-x(t)+\frac{1}{3}x^3(t)-y(t)),
        \end{cases}
    \end{equation*}
    where $d(t)$ follows a symmetric Beta distribution scaled to the interval $[-0.6,0.6]$, i.e., $d(t)=1.2\times \text{Beta}(20,20)-0.6$. The safe set is $\mathbb{X}=\{\,(x,y)^{\top}\mid x^2+y^2-1 \leq 0\,\}$.
\end{example}

In this example, we first apply the Robust-Set method to compute an RBF. Specifically, a polynomial template of degree $8$ is adopted for $h(\bm{a}, \bm{x})$. The resulting set $\mathbb{X}_s=\{\bm{x}\in \mathbb{X}\mid h(\bm{a}^*,\bm{x}) > 0\}$ is shown in Fig.~\ref{fig:bc_robust}, with an estimated volume of $0.4542$. While this yields a valid RBF, the result is relatively conservative since it enforces constraint \eqref{eq:rbf} to hold under all disturbances in $\mathbb{D}$. The proposed Data-Driven RBFs method mitigates this conservatism but does so at the cost of weakening the guarantees--from always-safety guarantees to PAC finite-step safety guarantees. For instance, when $\epsilon=0.01$, the Data-Driven RBFs method provides a lower bound of $0.9509$ for five-step safety with a confidence level of $99.9\%$. Moreover, as $\epsilon$ increases, the lower bound decreases. Fig.~\ref{fig:bc_robust} illustrates the resulting sets $\mathbb{X}_s$ for different values of $\epsilon$, along with three simulated trajectories starting from $(0.45, 0.75)$. The corresponding volumes and computation times are listed in
Table \ref{tab:vanderpol_epsilon}. 
Compared with the Robust-Set method, the data-driven approach enlarges $\mathbb{X}_s$, albeit at the cost of increased computation time. Furthermore, unlike the Robust-Set approach, our data-driven method does not require $\mathbb{D}$ to be explicitly known or semi-algebraic, thus enhancing its practical applicability.

\begin{table}[!htbp]
\captionsetup{skip=2pt}
\caption{Results of Data-Driven RBFs method for varying $\epsilon$.}
\setlength{\tabcolsep}{2mm}{
\begin{tabular}{cccccc}
\hline
$\epsilon$ & $V_{\mathbb{X}_s}$ & $T_{s}$ & $T_{c}$ & $T_{total}$  \\ \hline
0.5  & 0.8579 & 1.9  & 1.6  & 3.5\\
0.4  & 0.7385 & 1.9  & 1.9  & 3.7\\
0.3  & 0.7144 & 1.9  & 2.2  & 4.1\\
0.2  & 0.7388 & 2.0  & 3.3  & 5.3\\
0.1  & 0.6911 & 2.4  & 7.2  & 9.6\\
0.05 & 0.6652 & 3.3  & 25.3  & 28.6\\
0.01 & 0.6452 & 9.6  & 453.1  & 462.7 \\\hline
\end{tabular}}
\label{tab:vanderpol_epsilon}
\end{table}

We further examine the influence of polynomial degree on Data-Driven RBFs, as reported in Table~\ref{tab:vanderpol_rbf_degree}. Since the objective function $\frac{1}{N_o}\sum_{i=1}^{N_o}h(\bm{a},\bm{x}'_i)$ does not directly correspond to the volume of $\mathbb{X}_s$, a higher polynomial degree may improve the objective value but does not necessarily yield a larger $\mathbb{X}_s$. Moreover, although increasing the polynomial degree can enhance the expressive power of barrier functions, it also increases the dimension $m$ of $\bm{a}$, necessitating a larger sample size $M$ and thereby raising the computational burden. Since when all states in $\mathbb{X}$ are used for computations, the resulting set $\mathbb{X}_s$ coincides with the robust invariant set obtained via the Robust-Set method. Consequently, using higher-degree polynomials may lead to a smaller invariant set $\mathbb{X}_s$, as illustrated in Table \ref{tab:vanderpol_rbf_degree}. While this does not improve the accuracy of the PAC safety guarantees and, on the contrary, may lead to a smaller set $\mathbb{X}_s$, it could enhance robustness in certain practical applications.
\begin{table}[!htbp]
\captionsetup{skip=2pt}
\caption{Results for different degrees ($\epsilon = 0.2$).}
\setlength{\tabcolsep}{2mm}{
\begin{tabular}{ccccccc}
\hline
\multirow{2}{*}{Degree} & \multicolumn{2}{c}{Robust-Set} & \multicolumn{4}{c}{Data-Driven RBFs} \\ \cmidrule(lr){2-3}\cmidrule(lr){4-7}
& $V_{\mathbb{X}_s}$ & $T_{total}$ & $V_{\mathbb{X}_s}$ & $T_{s}$ & $T_{c}$ & $T_{total}$  \\ \hline
4 & 0.0000 & 2.4   & 0.7388 & 2.0  & 3.3  & 5.3\\
6 & 0.0000 & 2.8   & 0.7283 & 3.3  & 14.1 & 17.4\\
8 & 0.4542 & 3.8 & 0.7030 & 11.9 & 57.2 & 69.1 \\ \hline
\end{tabular}}
\label{tab:vanderpol_rbf_degree}
\end{table}

Next, we evaluate Data-Driven SBFs setting with $U_{\bm{a}} = 8$ and $\tau = 0.01$. Table~\ref{tab:vanderpol_sbf_degree} lists the results for polynomial degrees 4, 6, and 8, along with corresponding results from the Distribution-Known baseline. For the case where $h(\bm{a}, \bm{x})$ is of degree 8, Data-Driven SBFs method yields $\lambda^*(D^M)=0.0107$, and the contour of $1-\lambda^*(D^M)-h(\bm{a}^*(D^M), \bm{x})$ is shown in Fig.~\ref{fig:bc_stochastic}. For the Distribution-Known method, the corresponding value is $\lambda^*=0.0007$, with the contour shown in Fig.~\ref{fig:bc_known}. These results indicate that, regardless of the degree of $h(\bm{a}, \bm{x})$, the proposed Data-Driven SBFs achieves performance close to that of the ideal known-distribution case.
\begin{table}[!htbp]
\captionsetup{skip=2pt}
\caption{Results for different degrees.}
\setlength{\tabcolsep}{1.5mm}{
\begin{tabular}{ccccccc}
\hline
\multirow{2}{*}{Degree} & \multicolumn{2}{c}{Distribution-Known} & \multicolumn{4}{c}{Data-Driven SBFs} \\ \cmidrule(lr){2-3}\cmidrule(lr){4-7}
& $J^{*}$ & $T_{total}$ & $J^*$ & $T_{s}$ & $T_{c}$ & $T_{total}$  \\ \hline
4 & 0.5056 & 1.9 & 0.5155 & 6.1  & 0.6  & 6.7\\
6 & 0.2920 & 2.3 & 0.3019 & 6.3  & 0.6  & 6.9\\
8 & 0.1552 & 3.7 & 0.1651 & 12.6  & 1.2  & 13.8\\
\hline
\end{tabular}}
\label{tab:vanderpol_sbf_degree}
\end{table}

It's worth noting that in our Data-Driven RBFs method, the optimization step typically dominates the total computation time, whereas in our Data-Driven SBFs method, the optimization time is relatively small. This difference arises because the number of constraints in the optimization problem  \eqref{eq:RBF_linear21} grows with the number of samples, while the number of constraints in \eqref{eq:SBF_linear11} remains independent of the sample size. Thus, Data-Driven SBFs method is more computationally efficient for large-scale problems involving numerous samples, albeit at the cost of requiring more samples and providing only one-step safety guarantees.

\begin{figure*}[!htbp]
    \centering
    \subfigure[EX\ref{ex:bc}: Data-Driven RBFs and Robust-Set]{
    \includegraphics[width=0.19\linewidth]{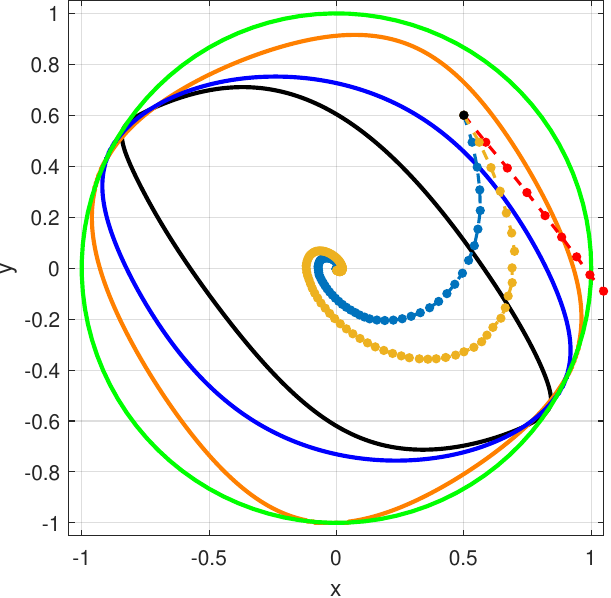}
    \label{fig:bc_robust}
    }
    \subfigure[EX\ref{ex:bc}: Data-Driven SBFs]{
    \includegraphics[width=0.19\linewidth]{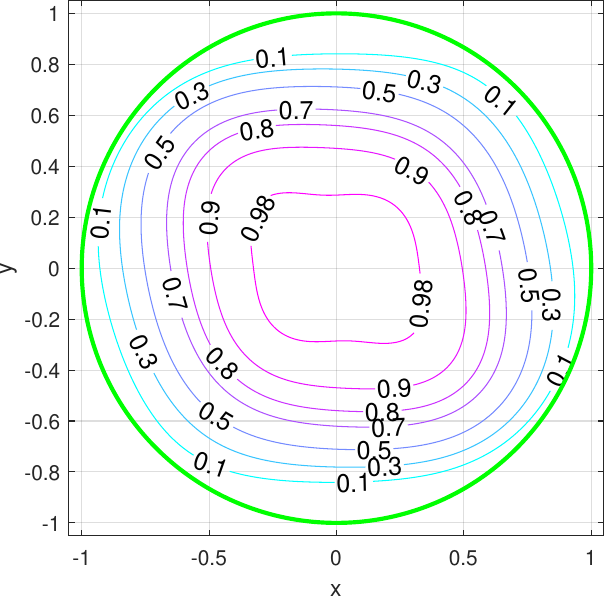}
    \label{fig:bc_stochastic}
    }
    \subfigure[EX\ref{ex:bc}: Distribution-Known]{
    \includegraphics[width=0.19\linewidth]{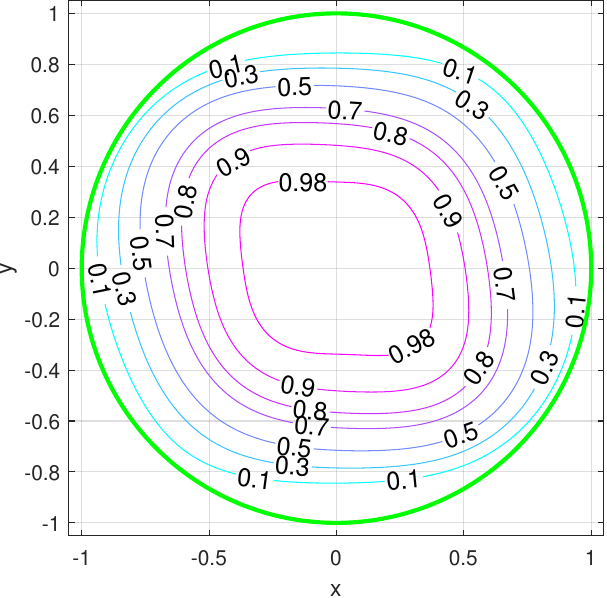}
    \label{fig:bc_known}
    }
    \subfigure[EX\ref{ex:lotka}: Data-Driven RBFs]{
    \includegraphics[width=0.19\linewidth]{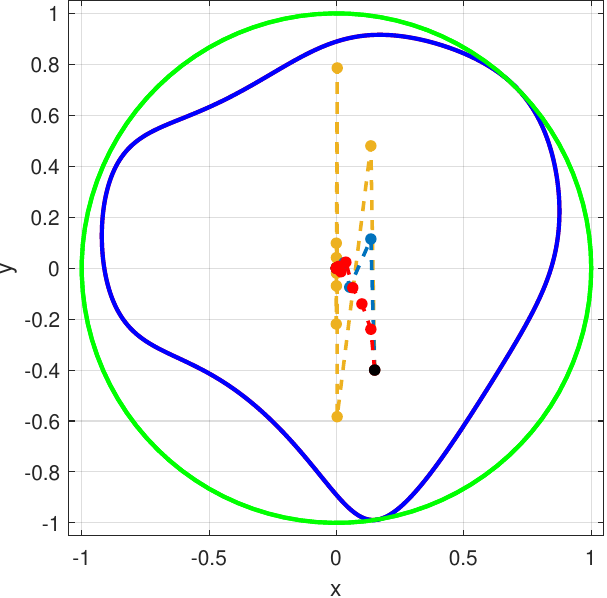}
    \label{fig:lotka_robust}
    }
    \subfigure[EX\ref{ex:lotka}: Data-Driven SBFs]{
    \includegraphics[width=0.19\linewidth]{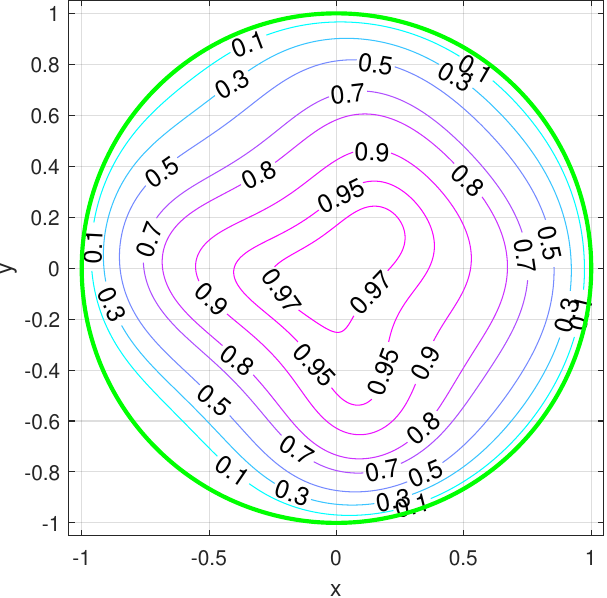}
    \label{fig:lotka_stochastic}
    }
    \subfigure[EX\ref{ex:lotka}: Distribution-Known]{
    \includegraphics[width=0.19\linewidth]{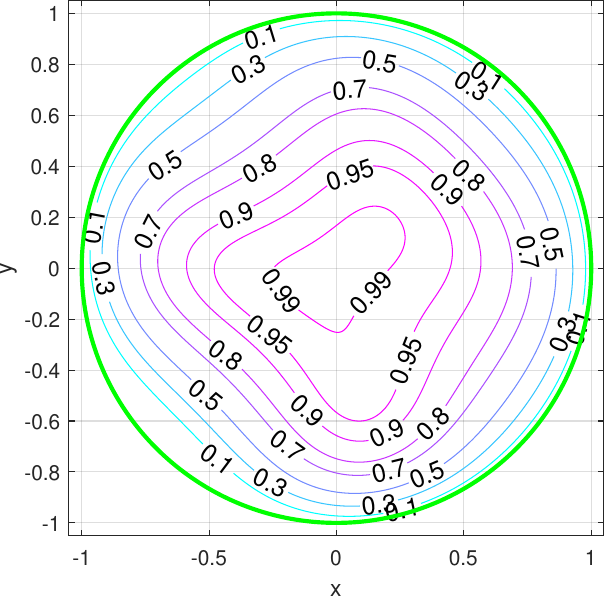}
    \label{fig:lotka_known}
    }
    \subfigure[EX\ref{ex:jet}: Data-Driven SBFs]{
    \includegraphics[width=0.19\linewidth]{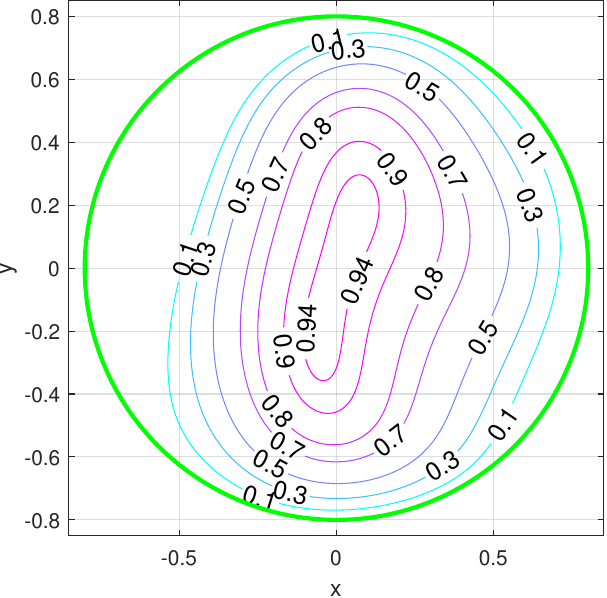}
    \label{fig:jet_ours}
    }
    \subfigure[EX\ref{ex:jet}: Distribution-Known]{
    \includegraphics[width=0.19\linewidth]{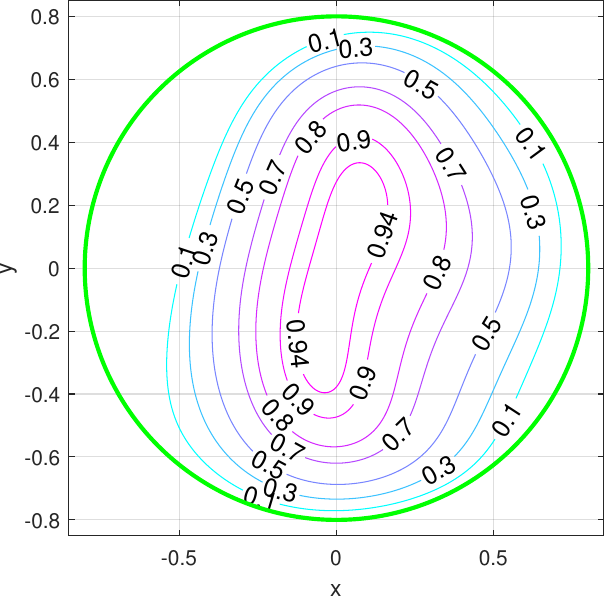}
    \label{fig:jet_known}
    }
    \caption{\textcolor{green}{Green} curve represents the boundary of the safe set $\mathbb{X}$. In Fig. \ref{fig:bc_robust}, black curve represents the $\mathbb{X}_s$ boundary obtained by the Robust-Set method, and \textcolor[rgb]{0 0.4470 0.7410}{blue} and \textcolor[rgb]{0.9290, 0.6940, 0.1250}{yellow} curves represent the $\mathbb{X}_s$ boundary obtained by Data-Driven RBFs when $\epsilon=0.01$ and $\epsilon=0.5$, respectively. Fig. \ref{fig:bc_robust} and \ref{fig:lotka_robust} plot three systems simulation trajectories starting from $(0.45, 0.75)$ and $(0.15, -0.4)$.}
    \label{fig:results}
\end{figure*}

\begin{example}
\label{ex:lotka}
    Consider the Lotka-Volterra model \cite{lotka2002contribution}:
\begin{equation*}
    \begin{cases}
x(t+1)=rx(t)-ay(t)x(t),\\
y(t+1)=sy(t)+acy(t)x(t),
\end{cases}
\end{equation*}
where $r=0.5$, $a=1$, $s=-0.5+d(t)$, and $\mathbb{X}=\{\,(x,y)^{\top}\mid x^2+y^2 \leq 1\,\}$. Besides, $d(t)$ follows a truncated normal distribution, defined on $[-1,1]$ with a mean of zero and a standard deviation of $0.15$.
\end{example}

In this example, the Robust-Set method fails to obtain a non-empty set $\mathbb{X}_s$. By contrast, the proposed Data-Driven RBFs method works for this case. 
The results for different polynomial degrees are reported in Table~\ref{tab:lotka_rbf_degree}, and Fig.~\ref{fig:lotka_robust} depicts the resulting set $\mathbb{X}_s$ for degree 6 along with three simulated trajectories initialized at $(0.15, -0.4)$. The results demonstrate that our method is able to compute a non-conservative $\mathbb{X}_s$, whereas the Robust-Set method yields $\mathbb{X}_s = \emptyset$. Moreover, contrary to Example \ref{ex:bc}, the size of $\mathbb{X}_s$ increases as the degree increases for this example.

\begin{table}[!htbp]
\captionsetup{skip=2pt}
\caption{Results for different degrees ($\epsilon = 0.1$).}
\setlength{\tabcolsep}{2mm}{
\begin{tabular}{ccccccc}
\hline
\multirow{2}{*}{Degree} & \multicolumn{2}{c}{Robust-Set} & \multicolumn{4}{c}{Data-Driven RBFs} \\ \cmidrule(lr){2-3}\cmidrule(lr){4-7}
& $V_{\mathbb{X}_s}$ & $T_{total}$ & $V_{\mathbb{X}_s}$ & $T_{s}$ & $T_{c}$ & $T_{total}$  \\ \hline
4 & 0.0000 & 2.3 & 0.0000 & 2.3  & 4.8  & 7.1\\
6 & 0.0000 & 2.7 & 0.7270 & 3.8  & 23.3 & 27.1\\
8 & 0.0000 & 3.4 & 0.7430 & 9.7  & 103.3  & 113.0 \\ \hline
\end{tabular}}
\label{tab:lotka_rbf_degree}
\end{table}

For the Data-Driven SBFs method, we set $U_{\bm{a}} = 8$, $\tau = 0.02$, and use a polynomial template of degree 6 for $h(\bm{a}, \bm{x})$. Our Data-Driven SBFs method yields $J^* = 0.5277$, and the contour of $1 - \lambda^*(D^M) - h(\bm{a}^*(D^M), \bm{x})$ is shown in Fig.~\ref{fig:jet_ours}. The total computation time is 52 seconds, of which 51 seconds are used for sampling and data preprocessing and 1 second for the optimization. For comparison, the Distribution-Known method produces $J^* = 0.5077$ in 12 seconds. These results demonstrate that even without knowledge of the true disturbance distribution, our Data-Driven SBFs method can generate reliable and non-conservative probabilistic safety guarantees based solely on sampled data.

\begin{table*}[!htb]
\captionsetup{skip=2pt}
\caption{Additional experimental results.}
\setlength{\tabcolsep}{2mm}{
\begin{tabular}{cccccccccccccc}
\hline
\multirow{2}{*}{Example} & \multirow{2}{*}{n} & \multicolumn{2}{c}{Robust-Set} & \multicolumn{4}{c}{Data-Driven RBFs} & \multicolumn{2}{c}{Distribution-Known} & \multicolumn{4}{c}{Data-Driven SBFs} \\ \cmidrule(lr){3-4} \cmidrule(lr){5-8} \cmidrule(lr){9-10} \cmidrule(lr){11-14} 
 & & $V_{\mathbb{X}_s}$ & $T_{total}$ & $V_{\mathbb{X}_s}$ & $T_{s}$ & $T_{c}$ & $T_{total}$ & $J^*$ & $T_{total}$ & $J^*$ & $T_{s}$ & $T_{c}$ & $T_{total}$ \\ \hline
\ref{ex:arch4} & 2 & 0.0000 & 3.5 & 0.8571 & 3.3 & 7.8 & 11.1 & 0.2044 & 5.2 & 0.2144 & 26.1 & 1.7 & 27.8 \\
\ref{ex:vinc}  & 2 & 0.0000 & 2.6 & 0.5286 & 2.6 & 7.3 & 9.9 & 0.3698 & 5.8 & 0.3798 & 8.4 & 0.9 & 9.3 \\
\ref{ex:bc4}   & 2 & 0.8312 & 2.2 & 0.8521 & 2.7 & 5.7 & 8.4 & 0.2251 & 5.8 & 0.2351 & 17.7 & 1.5 & 19.2 \\
\ref{ex:stable} & 3 & 0.7856 & 4.9 & 0.8633 & 9.6 & 95.0 & 104.6 & 0.3465 & 20.2 & 0.3565 & 31.2 & 1.0 & 32.2 \\
\ref{ex:3dvanderpol} & 3 & 0.0000 & 2.8 & 0.0000 & 5.8 & 316.1 & 321.9 & 0.1442 & 11.4 & 0.1542 & 8.0 & 5.3 & 13.3 \\
\ref{ex:sank}   & 4 & 0.0000 & 2.5 & 0.0000 & 2.2 & 12.8 & 15.0 & 0.3726 & 9.2 & 0.3826 & 6.6 & 4.5 & 11.1 \\
\ref{ex:lorenz} & 6 & 0.0000 & 2.7 & 0.9692 & 5.3 & 104.0 & 109.3 & 0.1274 & 65.3 & 0.1374 & 27.6 & 36.5 & 64.1 \\
\hline
\end{tabular}}
\label{tab:exp}
\end{table*}

In this example, the Data-Driven RBFs method provides a state-uniform lower bound on the $k$-step safety probability for all states contained in its computed set $\mathbb{X}_s$. For instance, when $\epsilon = 0.05$, the method not only certifies a uniform lower bound of approximately $0.95$ for one-step safety, but also a lower bound of $0.857$ for 3-step safety, both with $99.9\%$ confidence. However, these safety guarantees apply only to states within $\mathbb{X}_s$. In contrast, the Data-Driven SBFs method provides a state-wise probability lower bound for one-step safety only, while offering guarantees for every state in the entire safe set, as observed in Fig. \ref{fig:lotka_stochastic}. 
For instance, the Data-Driven SBFs method certifies a probability lower bound of $0.9573$ for one-step safety at $99.9\%$ confidence for the initial state $(0.15, -0.4) \in \mathbb{X}_s$. 
These results show that the Data-Driven RBFs and Data-Driven SBFs methods offer complementary strengths, and the choice between them should be guided by the characteristics and requirements of the specific problem at hand. Furthermore, both certification procedures can be applied recursively to establish safety guarantees over longer horizons.

\begin{example}
\label{ex:jet}
    Consider the discretization of Moore-Greitzer model of a jet engine \cite{aylward2008stability}:
\begin{equation*}
\begin{cases}
x(t+1)=x(t) + 0.1(-y(t)-1.5x^2(t)-0.5x^3(t)+d(t)),\\
y(t+1)=y(t) + 0.1(3x(t)-y(t)+d(t)),
\end{cases}
\end{equation*}
 where $d(\cdot)$ is uniformly distributed over $[-1.5,1.5]$. The safe set is $\mathbb{X}=\{\,(x,y)^{\top}\mid x^2+y^2 \leq 0.64\,\}$.
\end{example}


In this example, the Robust-Set method fails to produce a valid robust barrier certificate. In addition, for $\epsilon \le 0.99$, solving \eqref{eq:RBF_linear21} always produces $\mathbb{X}_s= \emptyset$, indicating that Theorems~\ref{lemma:robust} and \ref{thm:so} may not provide a meaningful guarantee.
  We therefore focus on the Data-Driven SBFs method. Setting $\tau = 0.01$, $U_{\bm a} = 30$, and using a polynomial template of degree 10 for $h(\bm{a}, \bm{x})$, solving \eqref{eq:SBF_linear11} yields $J^*=0.4819$ and $\lambda^*(D^M)=0.0479$, with a computation time of $51\ \mathrm{s}$. The contour of $1-\lambda^*(D^M)-h(\bm{a}^*(D^M),\bm{x})$ is shown in Fig.~\ref{fig:jet_ours}. Under the same settings but assuming the disturbance distribution is known, the optimal value is $J^*=0.4719$ and $\lambda^*=0.0379$ with a computation time of $19\ \mathrm{s}$, and the contour is shown in  Fig.~\ref{fig:jet_known}. The results from our method closely follow those obtained under the known distribution assumption, with our estimated bound being on average about $0.01$ lower, demonstrating the reliability of the proposed approach for one-step safety guarantees.


Fig. \ref{fig:jet_para} summarizes results under different hyperparameter settings. Increasing the polynomial degree of $h(\bm{a},\bm{x})$ improves expressiveness but requires longer computation. Increasing $U_{\bm{a}}$ generally yields a larger probability lower bound, while the improvement becomes moderate once $U_{\bm{a}}$ is sufficiently large. The effect of $\tau$ is approximately linear, and smaller $\tau$ results in a larger lower bound. However, as shown in \eqref{MM_RA}, the minimum required sample size $M$ grows with larger $U_{\bm{a}}$ or smaller $\tau$, which in turn increases computation time.
In practice, these hyperparameters can be tuned to trade off computational efficiency and certificate quality. Other parameters have negligible impact on both the computed results and the overall computation time.

\begin{figure}[!htbp]
    \centering
    \subfigure{
    \includegraphics[width=0.8\linewidth]{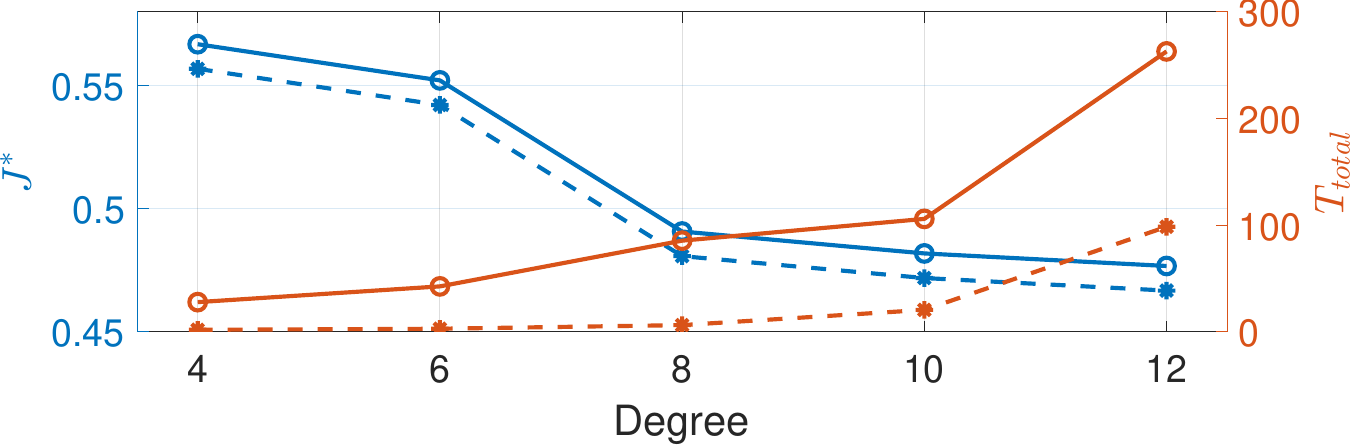}
    }
    \subfigure{
    \includegraphics[width=0.8\linewidth]{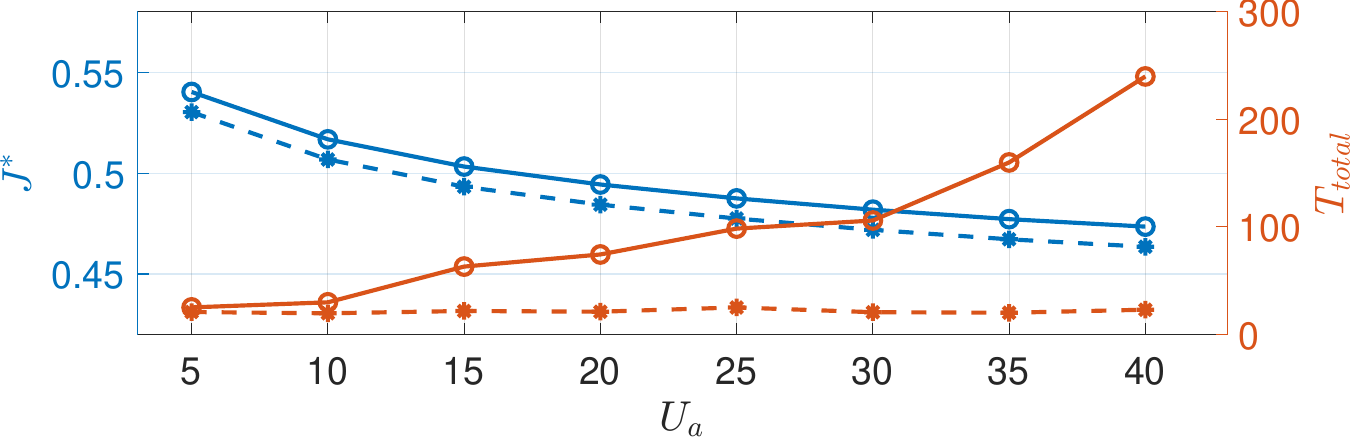}
    }
    \subfigure{
    \includegraphics[width=0.8\linewidth]{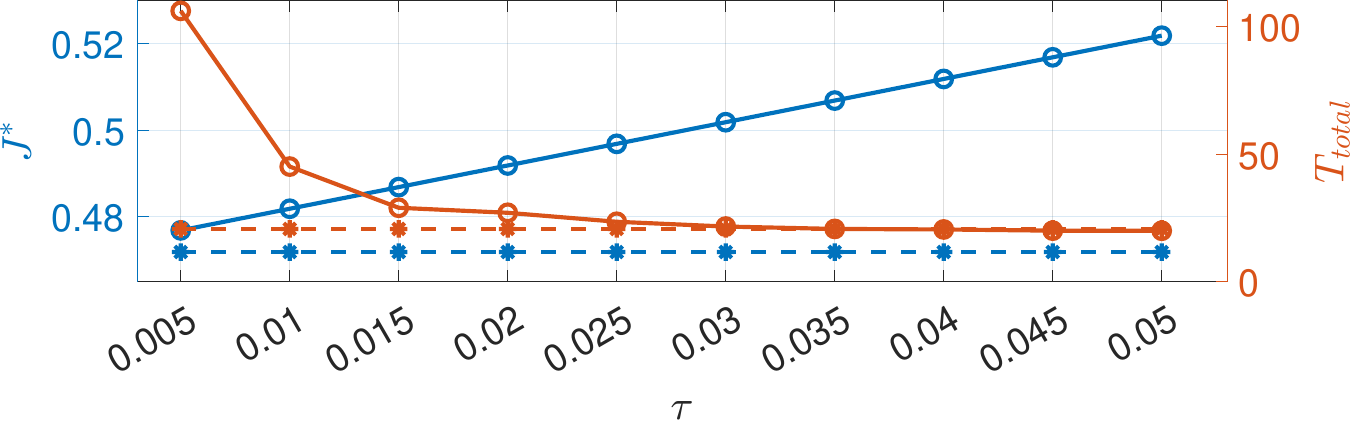}
    }
    \caption{Results under different hyperparameter. The solid line and dashed line correspond to the Data-Driven SBFs method and Distribution-Known method, respectively.}
    \label{fig:jet_para}
\end{figure}

The results of the remaining experiments are summarized in Table~\ref{tab:exp}, with the corresponding settings provided in Appendix~\ref{sec:app_exp}. In these experiments, we set $\epsilon=0.1$ for the proposed Data-Driven RBFs method. Compared with the Robust-Set method, Data-Driven RBFs method consistently yields larger sets $\mathbb{X}_s$ and remains applicable even when the Robust-Set method fails to produce a valid certificate. Meanwhile, the proposed Data-Driven SBFs achieves performance close to the Distribution-Known baseline. Both proposed methods demonstrate good scalability and are capable of handling six-dimensional systems. 

\section{Conclusion}
\label{sec:con}
This work studied PAC finite-time safety guarantees for stochastic discrete-time systems with unknown disturbance distributions. By reformulating robust and stochastic barrier-certificate conditions as empirical optimization problems, we established three complementary characterizations: one based on VC dimension, one using the scenario approach, and one leveraging Rademacher complexity. Together, these results connect finite-sample feasibility with probabilistic safety, enabling sample-efficient and distribution-free certification of stochastic systems.
Finally, we demonstrated the efficacy of the proposed approaches through several polynomial examples solved via semidefinite programming–based tools, illustrating both the practical performance and the inherent trade-offs of each approach. The results confirm that these data-driven methods could provide reliable safety certificates with high confidence, while highlighting the critical balance between conservatism, computational complexity, and sample efficiency. 

In this work, we consider PAC finite-time safety certification for discrete-time stochastic systems with known dynamics but unknown disturbance distributions (i.e., $\bm{f}(\bm{x},\bm{d})$ is known). {In many task-driven systems, however, not only safety but also successful reachability of a target set is essential. Accordingly, as a natural extension, we plan to generalize our framework to reach-avoid certification for such systems.} Further, in many applications, the dynamics themselves are also unknown. In future work, in addition to the directions discussed in Remark~\ref{relation}, we plan to extend PAC finite-time safety certification to systems with unknown dynamics and develop a certification procedure for the nested PAC guarantee as proposed in \cite{xue2020pac}. Also, we aim to extend the proposed approaches to PAC-style verification for probabilistic programs \cite{zaiser2025guaranteed} and quantum computations \cite{ying2010quantum}.
\oomit{
\begin{equation*}
\label{eq:pac_certification}
P_{\bm{x}}^N \left[ 
\begin{split}
    &P_{\bm{x}} \left[
    \begin{split}
        &P_{\bm{d}}^M \left[
        \mathcal{A}(S) \Rightarrow
        \begin{split}
            P_{\pi} \Big[&\bigwedge_{i=1}^k \bm{\phi}^{\bm{x}}_{\pi}(i) \in \mathbb{X} 
            \Big] \\
            &\ge 1-\epsilon_2(\bm{x})
            \end{split}
        \right] \\
        &\ge 1-\delta_2
        \end{split}
    \right] \\
    &\ge 1-\epsilon_1
    \end{split}
\right] \ge 1-\delta_1.
\end{equation*}
where $X^N =\{\bm{x}_i \in \mathbb{X}\}_{i=1}^N\stackrel{\text{i.i.d.}}{\sim} P_{\bm{x}}$ and $D_i^M=\{\bm{d}_{i,j} \in \mathbb{D}\}_{j=1}^M\stackrel{\text{i.i.d.}}{\sim} P_{\bm{d}}$ ($i=1,\ldots,N$).
That is, with confidence at least $1-\delta_1$, the probability measure of the set of states $\bm{x}$ in the safe set satisfying the following property is larger than or equal to $1-\epsilon_1$: 

\quad\big[with confidence at least $1-\delta_2$, if $\mathcal{A}(S)$ is accepted, the probability that the system remains within the safe set at the next time step is no less than a value determined by $\epsilon_2(\bm{x})$.\big]
}

\section*{Acknowledgments}
This work was partially supported by the National Research Foundation, Singapore, under its RSS Scheme (NRF-RSS2022-009), the CAS Pioneer Hundred Talents Program, and the Basic Research Program of the Institute of Software, Chinese Academy of Sciences (Grant No. ISCAS-JCMS-202302).

\bibliographystyle{ACM-Reference-Format}
\bibliography{sample-base}

\section*{Appendix}
\setcounter{secnumdepth}{2}
\setcounter{section}{0}
\renewcommand\thesubsection{\thesection.\arabic{subsection}}
\renewcommand\thesection{\Alph{section}}
\label{sec:appendix}

\section{Proofs in the Paper}
\subsection{The proof of Theorem \ref{lemma:robust}}
\label{sec:proof_theorem_robust}
\begin{proof}
Since $\textnormal{E}_{\bm{d},M}[\widehat{f}_{\bm{a}^*(D^M))}]=0, \forall\bm{x}\in \mathbb{X}$ holds, according to Lemma \ref{coro4}, we can obtain 
\begin{equation*}
\label{coro41}
\textnormal{P}_{\bm{d}}^M \big[\textnormal{E}_{\bm{d}}[\widehat{f}_{\bm{a}^*(D^M)}]\leq \epsilon\Big]
\geq 1-\delta, \forall \bm{x}\in \mathbb{X}.
\end{equation*}
\oomit{Let $\textnormal{P}_{\bm{d},M}$ the empirical measure on $M$ i.i.d.\ disturbances $\bm{d}^{(1)},\dots,\bm{d}^{(M)}$. For $\epsilon>0$ we seek a uniform one-sided bound
\begin{equation}
\label{vc}
\textnormal{P}_{\bm{d}}^M\Big[\sup_{\widehat{f}_{\bm{a},\xi}\in\mathbb{F}_{\bm{x}}}\big(\textnormal{P}_{\bm{d}}[\widehat{f}_{\bm{a},\xi}]-\textnormal{P}_{\bm{d},M}[\widehat{f}_{\bm{a},\xi}]\big)>\epsilon\Big]\leq\delta, \forall \bm{x}\in \mathbb{X}.
\end{equation}
A standard symmetrization + combinatorial argument (see, Theorem 2 in \cite{vapnik2015uniform}) gives, for any $M\geq \frac{2}{\epsilon^2}$,
\begin{equation}
\label{eq:symm}
\begin{split}
&\textnormal{P}_{\bm{d}}^M\Big[\sup_{\widehat{f}_{\bm{a},\xi}\in\mathbb{F}}\big(\textnormal{P}_{\bm{d}}[\widehat{f}_{\bm{a},\xi}]-\textnormal{P}_{\bm{d},M}[\widehat{f}_{\bm{a},\xi}]\big)>\epsilon\Big]
\\
\leq &2 \Pi_{\mathbb{F}_{\bm{x}}}(2M)\,\exp\!\Big(-\frac{M\epsilon^2}{8}\Big), \forall \bm{x}\in \mathbb{X},
\end{split}
\end{equation}
where $\Pi_{\mathbb{F}}N$ is the growth function. Under assumption \ref{ass:VC_finite}, by Sauer's lemma \cite{sauer1972density}, if $M\geq \frac{N}{2}$, we have
\begin{equation}
\label{eq:sauer}
\Pi_{\mathbb{F}_{\bm{x}}}(2M)\le\sum_{i=0}^{N} \binom{2M}{i}\le\Big(\frac{2eM}{m}\Big)^{N}, \forall \bm{x}\in \mathbb{X}.
\end{equation}
Substituting \eqref{eq:sauer} into \eqref{eq:symm} yields
\[
\begin{split}
&\textnormal{P}_{\bm{d}}^M\Big[\sup_{\widehat{f}_{\bm{a},\xi}\in\mathbb{F}_{\bm{x}}}\big(\textnormal{P}_{\bm{d}}[\widehat{f}_{\bm{a},\xi}]-\textnormal{P}_{\bm{d},M}[\widehat{f}_{\bm{a},\xi}]\big)>\epsilon\Big]\\
\le &2\Big(\frac{2eM}{N}\Big)^{N}\exp\!\Big(-\frac{M\epsilon^2}{8}\Big), \forall \bm{x}\in \mathbb{X}.
\end{split}
\]
Requiring the right-hand side $\le\delta$ is equivalent to the implicit condition
\begin{equation}
\label{eq:implicit}
\frac{M\epsilon^2}{8}\geq \ln 2 + N\ln\!\Big(\frac{2eM}{N}\Big)+\ln\frac{1}{\delta}.
\end{equation}
}

In addition, for each $\bm{x}\in \mathbb{X}$, we can obtain \[\wedge_{j=1}^M h(\bm{a}^*(D^M),\bm{f}(\bm{x},\bm{d}^{(j)}))\geq \gamma h(\bm{a}^*(D^M),\bm{x}).\]
Thus, we have
   \begin{equation*}
\begin{split}
  & \textnormal{P}_{\bm{d}}^M \left[\mathbb{X}_s\neq \emptyset \Rightarrow 
    \begin{split}
\Big( \textnormal{P}_{\bm{d}}\left[
\begin{split}
&h(\bm{a}^*(D^M),\bm{f}(\bm{x},\bm{d})) \\
& \geq\gamma h(\bm{a}^*(D^M),\bm{x})
\end{split}
\right]\geq 1-\epsilon\Big)
        \end{split}
        \right]\\   
          &\geq 
          1-\delta, \forall \bm{x}\in \mathbb{X}.
\end{split}
\end{equation*}

On the other hand, $h(\bm{a}^*(D^M),\bm{f}(\bm{x},\bm{d})) \geq\gamma h(\bm{a}^*(D^M),\bm{x})$ implies that 
if $\bm{x}\in \mathbb{X}_s$, $\bm{f}(\bm{x},\bm{d})\geq 0$ holds and thus $\bm{f}(\bm{x},\bm{d})\in \mathbb{X}_s$. Therefore, we have 
\begin{equation*}
\begin{split}
   &\textnormal{P}_{\bm{d}}^M \left[
    \begin{split}
        {\mathbb{X}_s\neq \emptyset} \Rightarrow  \textnormal{P}_{\bm{d}}[\bm{f}(\bm{x},\bm{d}) \in \mathbb{X}_s]\geq 1-\epsilon(\bm{x})
        \end{split}
        \right]\\
        &\geq 
          1-\delta, \forall \bm{x}\in \mathbb{X}.
\end{split}
\end{equation*}

The proof is finished.
\end{proof}

\subsection{The proof of Lemma \ref{lemma:robust_s}}
\label{sec:proof_lemma_robust_s}
\begin{proof}
Via treating $\bm{d} \in \mathbb{D}$ as uncertainty and $\bm{a}$ as decision variables, we first show that \eqref{eq:RBF_linear1} satisfies Assumption 1 in \cite{calafiore2006scenario}.   

Since $\mathbb{X}$ is compact and $\bm{f}(\bm{x},\bm{d})$ is continuous over $(\bm{x},\bm{d})$, we have $h(\bm{a}, \bm{f}(\bm{x},\bm{d})) \geq \gamma h(\bm{a}, \bm{x}), \forall \bm{x} \in \mathbb{X}$ is equivalent to 
\[\max_{\bm{x}\in \mathbb{X}}[-h(\bm{a}, \bm{f}(\bm{x},\bm{d}))+\gamma h(\bm{a}, \bm{x})]\leq 0\] 
under Assumption \ref{as:h}. Thus, we just show that \[\max_{\bm{x}\in \mathbb{X}}[-h(\bm{a}, \bm{f}(\bm{x},\bm{d}))+\gamma h(\bm{a}, \bm{x})]\] is continuous and convex over $\bm{a}$, for any fixed value of $\bm{d}\in \mathbb{D}$.


1) Since \begin{equation*}
\label{max_ineq}
\begin{split}
&\Big|\max_{\bm{x}\in \mathbb{X}}[-h(\bm{a}_1, \bm{f}(\bm{x},\bm{d}))+\gamma h(\bm{a}_1, \bm{x})]\\
&\quad\quad\quad\quad-\max_{\bm{x}\in \mathbb{X}}[-h(\bm{a}_2, \bm{f}(\bm{x},\bm{d}))+\gamma h(\bm{a}_2, \bm{x})]\Big|\\
\leq &\max_{\bm{x}\in \mathbb{X}}\Big|[-h(\bm{a}_1, \bm{f}(\bm{x},\bm{d}))+\gamma h(\bm{a}_1, \bm{x})]\\
&\quad\quad\quad\quad\quad\quad-[-h(\bm{a}_2, \bm{f}(\bm{x},\bm{d}))+\gamma h(\bm{a}_2, \bm{x})]\Big|,
\end{split}
\end{equation*}
according to Assumption \ref{as:h}, we can obtain 
\[
\begin{split}
\lim_{\bm{a}_1\rightarrow \bm{a}_2}\max_{\bm{x}\in \mathbb{X}}&[-h(\bm{a}_1, \bm{f}(\bm{x},\bm{d}))+\gamma h(\bm{a}_1, \bm{x})]\\
&=\max_{\bm{x}\in \mathbb{X}}[-h(\bm{a}_2, \bm{f}(\bm{x},\bm{d}))+\gamma h(\bm{a}_2, \bm{x})].
\end{split}
\]
This conclusion follows from the facts that $h(\bm{a},\bm{x})$ is continuous over $(\bm{a},\bm{x})$, $\bm{f}(\bm{x},\bm{d})$ is continuous over $(\bm{x},\bm{d})$, and $\mathbb{X}$ is compact. Thus, $\max_{\bm{x}\in \mathbb{X}}[-h(\bm{a}, \bm{f}(\bm{x},\bm{d}))+\gamma h(\bm{a}, \bm{x})]$ is continuous over $\bm{a}$. 

2) Assume that $h_1(\bm{\alpha}_1):=\max_{\bm{x}\in \mathbb{X}}[-h(\bm{a}_1, \bm{f}(\bm{x},\bm{d}))+\gamma h(\bm{a}_1, \bm{x})]$ and $h_2(\bm{\alpha}_2):=\max_{\bm{x}\in \mathbb{X}}[-h(\bm{a}_2, \bm{f}(\bm{x},\bm{d}))+\gamma h(\bm{a}_2, \bm{x})]$, we will show 
$\max_{\bm{x}\in \mathbb{X}}[-h(\bm{c}, \bm{f}(\bm{x},\bm{d}))+\gamma h(\bm{c}, \bm{x})]\leq \lambda h_1(\bm{\alpha}_1)+(1-\lambda)h_2(\bm{\alpha}_2)$, where 
$\bm{c}=\lambda \bm{a}_1+(1-\lambda) \bm{a}_2$ with $\lambda\in [0,1]$. 

According to Assumption \ref{as:h}, we have
\[
\begin{split}
&h(\bm{c}, \bm{f}(\bm{x},\bm{d}))-\gamma h(\bm{c}, \bm{x})\\
&=\lambda (h(\bm{a}_1, \bm{f}(\bm{x},\bm{d}))-\gamma h(\bm{a}_1, \bm{x}))\\
&\quad\quad\quad\quad\quad +(1-\lambda)(h(\bm{a}_2, \bm{f}(\bm{x},\bm{d}))-\gamma h(\bm{a}_2, \bm{x})).
\end{split}
\]
Thus, 
\[
\begin{split}
&\max_{\bm{x}\in \mathbb{X}}[-h(\bm{c}, \bm{f}(\bm{x},\bm{d}))+\gamma h(\bm{c}, \bm{x})]
\leq \lambda h_1(\bm{\alpha}_1)+(1-\lambda)h_2(\bm{\alpha}_2).
\end{split}
\]

Consequently, $\max_{\bm{x}\in \mathbb{X}}[-h(\bm{a}, \bm{f}(\bm{x},\bm{d}))+\gamma h(\bm{a}, \bm{x})]$ is convex over $\bm{a}$.
Similarly, we can show $\max_{\bm{x}\in \overline{\widehat{\mathbb{X}}\setminus \mathbb{X}}}h(\bm{a},\bm{x})$ is convex and continuous over $\bm{a}$. Thus,  \eqref{eq:RBF_linear1} satisfies Assumption 1 in \cite{calafiore2006scenario}.  

In addition, as mentioned before, there always exist solutions to \eqref{eq:RBF_linear21}. $\bm{a}=\bm{0}$ is a solution. \eqref{eq:RBF_linear21} satisfies Assumption 2 in \cite{calafiore2006scenario}. Thus, according to Theorem 1 in \cite{calafiore2006scenario}, we have 
 \begin{equation*}
\begin{split}
  & \textnormal{P}_{\bm{d}}^M \left[{\mathbb{X}_s\neq \emptyset} \Rightarrow 
    \begin{split}
\Big(\textnormal{P}_{\bm{d}}\left[
\begin{split}
&\max_{\bm{x}\in \mathbb{X}}[-h(\bm{a}^*(D^M),\bm{f}(\bm{x},\bm{d})) \\
& +\gamma h(\bm{a}^*(D^M),\bm{x})]\leq 0
\end{split}
\right]\geq 1-\epsilon\Big)
        \end{split}
        \right]\\   
          &\geq 
          1-\delta,
\end{split}
\end{equation*}
implying
  \begin{equation*}
\begin{split}
  & \textnormal{P}_{\bm{d}}^M \left[{\mathbb{X}_s\neq \emptyset}  \Rightarrow 
    \begin{split}
\Big(\forall \bm{x}\in \mathbb{X}, \textnormal{P}_{\bm{d}}\left[
\begin{split}
&h(\bm{a}^*(D^M),\bm{f}(\bm{x},\bm{d})) \\
& \geq\gamma h(\bm{a}^*(D^M),\bm{x})
\end{split}
\right]\geq 1-\epsilon\Big)
        \end{split}
        \right]\\   
          &\geq 
          1-\delta,
\end{split}
\end{equation*}
and 
\begin{equation*}
\begin{split}
   &\textnormal{P}_{\bm{d}}^M \left[
    \begin{split}
         {\mathbb{X}_s\neq \emptyset} \Rightarrow \Big(\forall \bm{x}\in \mathbb{X}_s, \textnormal{P}_{\bm{d}}[\bm{f}(\bm{x},\bm{d}) \in \mathbb{X}_s]\geq 1-\epsilon\Big)
        \end{split}
        \right]\\
        &\geq 
          1-\delta.
\end{split}
\end{equation*}
The proof is finished.
\end{proof}

\subsection{The proof of Theorem \ref{thm:so}}
\label{sec:proof_thm_so}
\begin{proof}
By Lemma~\ref{lemma:robust_s} there exists at least $1-\delta$-probability event over the sample $D^M$ on which the one-step safety bound holds uniformly over $\mathbb{X}_s$. Define the uniform event
\[
\mathbb{E}: =\Big\{D^M \mid  \forall \bm{y}\in\mathbb{X}_s,\textnormal{P}_{\bm d} \big[\bm{f}(\bm{y},\bm{d})\in\mathbb{X}_s\big]\ge 1-\epsilon\Big\}.
\]
Lemma~\ref{lemma:robust_s} guarantees $\textnormal{P}_{\bm{d}}^M(\mathbb{E})\ge 1-\delta$. 

Condition on any realization $D^M\in\mathbb{E}$. Fix an arbitrary initial state $\bm{x}\in\mathbb{X}_s$ and denote the trajectory by $\bm{x}_t:=\bm{\phi}_{\pi}^{\bm {x}}(t)$. By the law of total probability (tower property) we have the exact identity
\[
\begin{split}
&\textnormal{P}_{\pi}\big[\bm{x}_{t+1}\in\mathbb{X}_s \mid \bm{x}_t\in\mathbb{X}_s\big]\\
= ~&\int_{\bm{y}\in\mathbb{X}_s} \textnormal{P}_{\bm{d}}\big[\bm{f}({y},\bm {d})\in\mathbb{X}_s\big]\,d \textnormal{P}_{\bm{x}_t\mid \bm{x}_t\in\mathbb{X}_s}[\bm{y}]\\
= ~&\textnormal{E}_{\bm{y}\sim \textnormal{P}_{\bm{x}_t\mid \bm{x}_t\in\mathbb{X}_s}}\big[\textnormal{P}_{\bm{d}}[\bm{f}(\bm{y},\bm{d})\in\mathbb{X}_s]\big],
\end{split}
\]
where $\textnormal{P}_{\bm{x}_t\mid \bm{x}_t\in\mathbb{X}_s}$ refers to the conditional distribution of the random variable $\bm{x}_t$  given that it belongs to $\mathbb{X}_s$. Formally, for any measurable subset $\mathbb{B} \subseteq \mathbb{X}_s$, 
$
\textnormal{P}_{\bm{x} \mid \bm{x} \in \mathbb{X}_s}[\bm{x} \in \mathbb{B}]
= \frac{\textnormal{P}_{\bm{x}}[\bm{x} \in \mathbb{B} \cap \mathbb{X}_s]}{\textnormal{P}_{\bm{x}}[\bm{x} \in \mathbb{X}_s]}$ (
The condition $\textnormal{P}_{\bm{x}}[\bm{x} \in \mathbb{X}_s] > 0$ is required for the conditional probability 
$\textnormal{P}_{\bm{x}_t\mid \bm{x}_t\in\mathbb{X}_s}$ to be well-defined.
In our method, we take $\textnormal{P}_{\bm{x}}$ to be a Lebesgue-type measure so that any nonempty open set has positive measure. Thus, this condition $\textnormal{P}_{\bm{x}}[\bm{x} \in \mathbb{X}_s] > 0$ is satisfied by construction.
Since $h(\bm{a}^*(D^M),\bm{x})$ is continuous with respect to $\bm{x}$, 
if $h(\bm{a}^*(D^M),\bm{x}) > 0$ at some point $\bm{x}$, 
then there exists a neighbourhood $\Delta$ of $\bm{x}$ with positive (Lebesgue) measure such that
\[
h(\bm{a}^*(D^M),\bm{y}) > 0, 
\quad \forall\, \bm{y} \in \Delta.
\]
Consequently, $\Delta \subseteq \mathbb{X}_s$ and $\textnormal{P}_{\bm{x}}[\Delta] > 0$. 
Therefore,
\[
\textnormal{P}_{\bm{x}}[x_t \in \mathbb{X}_s] \geq \textnormal{P}_{\bm{x}}[x_t \in \Delta] >0,
\]
which ensures that $\textnormal{P}_{\bm{x}}[\bm{x}\in \mathbb{X}_s] > 0$ holds automatically under the continuity of $h(\bm{a}^*(D^M),\bm{x})$ ).

Under $D^M\in\mathbb{E}$ the integrand satisfies $\textnormal{P}_{\bm{d}}[\bm{f}(\bm{y},\bm{d})\in\mathbb{X}_s]\ge 1-\epsilon$ for every $\bm{y}\in\mathbb{X}_s$, hence the expectation (an average of values $\ge 1-\epsilon$) is itself at least $1-\epsilon$. Therefore, for every time $t$,
\[
\textnormal{P}_{\pi}\big[\bm{x}_{t+1}\in\mathbb{X}_s \mid \bm{x}_t\in\mathbb{X}_s\big] \ge 1-\epsilon.
\]

Applying the Markov property to the $k$-step conjunction yields
\[
\textnormal{P}_{\pi}\Big[\bigwedge_{i=1}^k \bm{x}_i\in\mathbb{X}_s \mid \bm{x}_0=\bm{x} \Big]
= \prod_{i=0}^{k-1} \textnormal{P}_{\pi}\big[\bm{x}_{i+1}\in\mathbb{X}_s \mid \bm{x}_i\in\mathbb{X}_s\big]
\ge (1-\epsilon)^k.
\]

The preceding inequalities were derived under the conditioning $D^M\in\mathbb{E}$. Removing that conditioning and using $\textnormal{P}_{\bm{d}}^M(\mathbb{E})\ge 1-\delta$ gives the stated high-probability guarantee:
\[
\begin{split}
&\textnormal{P}_{\bm{d}}^M\Big[\,{\mathbb{X}_s\neq \emptyset} \Rightarrow
\big(\forall \bm{x}\in\mathbb{X}_s,\; 
\textnormal{P}_{\pi}[\wedge_{i=1}^k \bm{\phi}_{\pi}^{\bm{x}}(i)\in\mathbb{X}_s]\ge (1-\epsilon)^k\big)
\Big] \\
&\ge 1-\delta,
\end{split}
\]
which completes the proof.
\end{proof}

\subsection{The proof of Theorem \ref{lemma:1}}
\label{sec:proof_lemma_1}
\begin{proof}

The empirical Rademacher complexity  is
\[
\Re_M(\mathbb{G}_{\bm{x}})=\frac{1}{M}\textnormal{E}_{\sigma}\Big[\sup_{\|\bm{a}\|_2\leq U_a}\big|\sum_{j=1}^M \sigma_jh(\bm{a},\bm{f}(\bm{x},\bm{d}^{(j)}))\big|\Big]
\]
for any $\bm{x}\in \mathbb{X}$, where $\sigma_j$ are independent Rademacher random variables ($\textnormal{P}[\sigma_i=+1]=\textnormal{P}[\sigma_i=-1]=\frac{1}{2}$) \cite{shalev2014understanding}. 
By standard results for linear classes with $\|\bm{\phi}(\bm{x},\bm{d})\|_2\leq R$, $\forall (\bm{x},\bm{d})\in \mathbb{X}\times \mathbb{D}$, we obtain 
\[
\Re_M(\mathbb{G}_{\bm{x}})\leq \frac{U_a R}{\sqrt{M}}, \forall \bm{x}\in \mathbb{X}.
\]

By the one-sided Rademacher concentration inequality \cite{mitzenmacher2017probability}, we conclude that, for any $t>0$,
\[
\begin{split}
&\textnormal{P}_{\bm{d}}^M\left[ 
\begin{split}
&\sup_{\bm{a}\in [-U_a,U_a]^m}\big( \textnormal{E}_{\bm{d}}[h(\bm{a},\bm{f}(\bm{x},\bm{d}))] - \tfrac{1}{M} \sum_{j=1}^{M} h(\bm{a}, \bm{f}(\bm{x}, \bm{d}^{(j)}))\big)\\ &
> 2\Re_M(\mathbb{G}_{\bm{x}})+t
\end{split}
\right]\\
&\leq e^{-\frac{M t^2}{2U_a^2}},  \forall \bm{x}\in \mathbb{X}.
\end{split}
\]

Choose $t=U_a\sqrt{\tfrac{2}{M}\ln\tfrac{1}{\delta}}$, so the deviation is
\[
2\Re_M(\mathbb{G}_{\bm{x}})+t \le \tau, \forall \bm{x}\in \mathbb{X}.
\]
That is,
\[
\frac{2U_a R}{\sqrt{M}} + U_a\sqrt{\tfrac{2}{M}\ln\tfrac{1}{\delta}} \le \tau, \forall \bm{x}\in \mathbb{X}.
\]

Let $s=\sqrt{M}$. Multiplying both sides by $s/U_a$, we get
\[
2R + \sqrt{2\ln\tfrac{1}{\delta}} \le \tfrac{\tau}{U_a}s,
\]
which yields the condition
\[
M = s^2 \ge \left( \frac{2RU_a + U_a\sqrt{2\ln(1/\delta)}}
             {\tau}\right)^2.
\]
        
Thus, we can obtain
\begin{equation*}
    \begin{split}
    \textnormal{P}_{\bm{d}}^M \left[
    \begin{split}
        &\textnormal{E}_{\bm{d}}[h(\bm{a}^*(D^M), \bm{f}(\bm{x},\bm{d}))]\leq \\
        &\tfrac{1}{M}\sum_{j=1}^M h(\bm{a}^*(D^M), \bm{f}(\bm{x}, \bm{d}^{(j)})) + \tau
        \end{split}
    \right]\geq 1-\delta, \forall \bm{x}\in \mathbb{X}.
    \end{split}
    \end{equation*}

    On the other hand, if \[
    \begin{split}
    &\textnormal{E}_{\bm{d}}[h(\bm{a}^*(D^M), \bm{f}(\bm{x},\bm{d}))]\leq \tfrac{1}{M}\sum_{j=1}^M h(\bm{a}^*(D^M), \bm{f}(\bm{x}, \bm{d}^{(j)})) +\tau
    \end{split}
    \] holds, then
        \[\textnormal{E}_{\bm{d}}[h(\bm{a}^*(D^M), \bm{f}(\bm{x},\bm{d}))]\leq h(\bm{a}^*(D^M),\bm{x})+\lambda^*(D^M)\] holds from \eqref{eq:SBF_linear11}. Further, under the condition that $h(\bm{a}^*(D^M),\bm{x}) \geq 0$ over $\mathbb{X}$ and $h(\bm{a}^*(D^M),\bm{x}) \geq 1$ over $\widehat{\mathbb{X}}\setminus \mathbb{X}$, we have 
\[\textnormal{P}_{\bm{d}} \big[\bm{f}(\bm{x},\bm{d})\in \mathbb{X}\big] 
        \geq  1-\lambda^*(D^M)-h(\bm{a}^*(D^M), \bm{x}).\]
        
Consequently, we obtain
\[
\begin{split}&
\textnormal{P}_{\bm{d}}^{M} \left[
\begin{split}
\textnormal{P}_{\bm{d}} [\bm{f}(\bm{x},\bm{d})\in \mathbb{X}] \geq 1-\lambda^*(D^M)-h(\bm{a}^*(D^M), \bm{x})
\end{split}
\right]\\
&\geq 1-\delta, \forall \bm{x}\in \mathbb{X}
\end{split}
\]
and finish the proof.
\end{proof}

\section{Implementation Details of the SOS Programming}
\label{sec:app_sos}
In this subsection, we present the implementation details of the SOS programming. Without loss of generality, let $\mathbb{X} = \{\bm{x}\in\mathbb{R}^n\mid g(\bm{x}) < 0\}$ and $\widehat{\mathbb{X}} = \{\bm{x}\in\mathbb{R}^n\mid \hat{g}(\bm{x}) < 0\}$, where $g(\bm{x})$ and $\hat{g}(\bm{x})$ are polynomial functions. The SOS programmings corresponding to optimization problem \eqref{eq:RBF_linear21} and \eqref{eq:SBF_linear11} are formulated in \eqref{eq:rbc_sos2} and \eqref{eq:sbc_sos1}, respectively. The proposed Data-Driven RBFs method involves solving \eqref{eq:rbc_sos2}, while the Data-Driven SBFs method requires solving \eqref{eq:sbc_sos1}. For comparison, when the disturbance set is given as $\mathbb{D} = \{\bm{d}\in\mathbb{R}^d \mid p(\bm{d}) \leq 0\}$, where $p(\bm{d})$ is a known polynomial function, \eqref{eq:rbf} can be recast into \eqref{eq:rbc_sos}, i.e., the Robust-Set method. When both the $\mathbb{D}$ and the disturbance distribution $\mathbb{P}_d$ are assumed to be known, $\eqref{eq:sbf}$ can be recast into $\eqref{eq:sbc_sos}$, corresponding to the Distribution-Known method.

In the sequel, $\sum[\bm{y}]$ denotes the set of sum-of-squares polynomials over variables $\bm{y}$, i.e., 
\[\sum[\bm{y}]=\{p\in \mathbb{R}[\bm{y}]\mid p=\sum_{i=1}^k q^2_i(\bm{y}), q_i(\bm{y})\in \mathbb{R}[\bm{y}],i=1,\ldots,k\},\] where $\mathbb{R}[\bm{y}]$ denotes the ring of polynomials in variables $\bm{y}$.

\begin{align}
\label{eq:rbc_sos2}
    &\min - \textstyle\frac{1}{N_o}\textstyle\sum_{k=1}^{N_o}h(\bm{a},\bm{x}'_k)\\
    \text{s.t.}&\begin{cases}
     h(\bm{a}, \bm{f}(\bm{x},\bm{d}^{(j)})) - \gamma h(\bm{a}, \bm{x}) + s_0(\bm{x}) g(\bm{x}) \in \sum [\bm{x}], \\
        -h(\bm{a},\bm{x}) + s_{1}(\bm{x}) \hat{g}(\bm{x}) - s_{2}(\bm{x}) g(\bm{x}) \in \sum [\bm{x}],\\
        s_{0}(\bm{x}), s_1(\bm{x}),s_2(\bm{x}) \in \sum [\bm{x}],\\
        \bm{a}[l]  \in [-U_a, U_a], \\
        ~ j=1,\ldots,M;~ l=1,\ldots,m, \\
    \end{cases}\notag
\end{align}

\begin{align}
    \label{eq:sbc_sos1}
    &\min \lambda + \textstyle\frac{1}{N_o}\textstyle\sum_{k=1}^{N_o}h(\bm{a},\bm{x}'_k)\\
    \text{s.t.}&
    \begin{cases}
       h(\bm{a},\bm{x}) + s_1(\bm{x}) g(\bm{x}) \in \sum [\bm{x}],\\
       h(\bm{a},\bm{x}) - 1 + s_2(\bm{x}) \hat{g}(\bm{x}) -s_3(\bm{x}) g(\bm{x}) \in \sum[\bm{x}],\\
       -\hat{h}(\bm{x},\bm{a},\lambda) + s_4(\bm{x}) g(\bm{x}) \in \sum [\bm{x}],\\
       s_1(\bm{x}),s_2(\bm{x}),s_3(\bm{x}),s_4(\bm{x}) \in \sum [\bm{x}],\\
        \lambda \in [0,1]; ~ \bm{a}[l] \in [-U_a, U_a]; ~ l=1,\ldots,m,
    \end{cases}\notag
\end{align}

\begin{align}
\label{eq:rbc_sos}
    &\min - \textstyle\frac{1}{N_o}\textstyle\sum_{k=1}^{N_o}h(\bm{a},\bm{x}'_k)\\
    \text{s.t.}&\begin{cases}
     h(\bm{a}, \bm{f}(\bm{x},\bm{d})) - \gamma h(\bm{a}, \bm{x}) + s_1(\bm{x}, \bm{d}) g(\bm{x}) \\
     \quad\quad\quad\quad\quad\quad\quad\quad\quad\quad\quad+ s_2(\bm{x}, \bm{d}) p(\bm{d}) \in \sum [\bm{x},\bm{d}], \\
        -h(\bm{a},\bm{x}) + s_3(\bm{x}) \hat{g}(\bm{x}) - s_4(\bm{x}) g(\bm{x}) \in \sum [\bm{x}],\\
        s_1(\bm{x}, \bm{d}),s_2(\bm{x},\bm{d}) \in \sum [\bm{x},\bm{d}]; s_3(\bm{x}),s_4(\bm{x}) \in \sum [\bm{x}],\\
        \|\bm{a}\|_2 \leq U_a; ~ l=1,\ldots,m, \\
    \end{cases}\notag
\end{align}

\begin{align}
    \label{eq:sbc_sos}
    &\min \lambda + \textstyle\frac{1}{N_o}\textstyle\sum_{k=1}^{N_o}h(\bm{a},\bm{x}'_k)\\
    \text{s.t.}&
    \begin{cases}
       h(\bm{a},\bm{x}) + s_1(\bm{x}) g(\bm{x}) \in \sum [\bm{x}],\\
       h(\bm{a},\bm{x}) - 1 + s_2(\bm{x}) \hat{g}(\bm{x}) -s_3(\bm{x}) g(\bm{x}) \in \sum[\bm{x}],\\
       h(\bm{x}) + \lambda -\mathbb{E}_{\bm{d}}[h(\bm{f}(\bm{x},\bm{d}))] + s_4(\bm{x}) g(\bm{x}) \in \sum [\bm{x}],\\
       s_1(\bm{x}),s_2(\bm{x}),s_3(\bm{x}),s_4(\bm{x}) \in \sum [\bm{x}],\\
        \lambda \in [0,1]; ~ \|\bm{a}\|_2 \leq U_a; ~ l=1,\ldots,m,
    \end{cases}\notag
\end{align}

\section{Additional Examples}
\label{sec:app_exp}
In this subsection, we provide additional examples with specific settings.

\begin{example}
\label{ex:arch4}
    Consider the following system adapted from \cite{sogokon2016non}:
\begin{equation*}
    \begin{cases}
x(t+1)=x(t)+0.01\big(-2x(t)+x^2(t)+y(t)\big),\\
y(t+1)=y(t)+0.01\big(x(t)-2y(t)+y^2(t)+d(t)\big),
\end{cases}
\end{equation*}
where $d(t)$ follows a truncated normal distribution, defined on $[-1,1]$ with a mean of zero and a standard deviation of $0.15$. The safe set is $\mathbb{X}=\{\,(x,y)^{\top}\mid x^2+y^2 \leq 1\,\}$.
\end{example}

\begin{example}
    \label{ex:vinc}Consider the following system adapted from \cite{vincent1997nonlinear}:
\begin{equation*}
    \begin{cases}
x(t+1)=x(t)+0.01\big(y(t) - x(t) d(t)\big),\\
y(t+1)=y(t)+0.01\big(-(1-x^2(t))x(t) - y(t)\big),
\end{cases}
\end{equation*}
where $d(t)$ follows a symmetric Beta distribution scaled to the interval $[-2,2]$, i.e., $d(t)=4\times \text{Beta}(250,250)-2$.. The safe set is $\mathbb{X}=\{\,(x,y)^{\top}\mid x^2+y^2 \leq 0.64\,\}$. 
\end{example}

\begin{example}
    \label{ex:bc4}
    Consider the following system adapted from \cite{zhang2018safety}:
\begin{equation*}
    \begin{cases}
x(t+1)=x(t)+0.01\big(-x(t)+2x^2(t)y(t)+x(t)d_1(t)\big),\\
y(t+1)=y(t)+0.01\big(-y(t)+y(t)d_2(t)\big),
\end{cases}
\end{equation*}
where $d_1(t)$ and $d_2(t)$ are independent and follow the uniform distribution on the interval $[-0.5, 0.5]$. The safe set is $\mathbb{X}=\{\,(x,y)^{\top}\mid x^2+y^2\leq 1 \,\}$.
\end{example}

\begin{example}
\label{ex:stable}
    Consider the following system adapted from \cite{ben2015stability}:
\begin{equation*}
\begin{cases}
x(t+1)=x(t)+\tau_s\big(-x(t)+y(t)-z(t)-d_1(t)\big),\\
y(t+1)=y(t)+\tau_s\big(-x(t)(z(t)+1)-y(t)-d_2(t)\big),\\
z(t+1)=z(t)+\tau_s\big(0.76524x(t) -4.7037z(t)-d_3(t)\big),\\
\end{cases}
\end{equation*}
where $d_1(t)$, $d_2(t)$, and $d_3(t)$ are independent and follow uniform distributions on $[-0.5, 0.5]$. The safe set is $\mathbb{X}=\{\,(x,y,z)^{\top}\mid x^2+y^2+z^2 \leq 1\,\}$, and $\tau_s=0.01$. 
\end{example}

\begin{example}
\label{ex:3dvanderpol}
    Consider the following system adapted from \cite{korda2014controller}:
\begin{equation*}
\begin{cases}
x(t+1)=x(t)+\tau_s\big(-2y(t)\big),\\
y(t+1)=y(t)+\tau_s\big(0.8x(t)-2.1y(t)+z(t)+10x^2(t)y(t)\big),\\
z(t+1)=z(t)+\tau_s\big(-z(t)+z^3(t)+d(t)\big),\\
\end{cases}
\end{equation*}
where $d(t)$ follows a uniform distribution on $[-0.5, 0.5]$. The safe set is $\mathbb{X}=\{\,(x,y,z)^{\top}\mid x^2+y^2+z^2 \leq 1\,\}$, and $\tau_s=0.01$. 
\end{example}

\begin{example}
    \label{ex:sank} Consider a model adapted from \cite{sankaranarayanan2013lyapunov} with a discrete time $0.01$,
    \begin{equation*}
        \begin{cases}
        x_1(t+1)=x_1(t) + \tau_s(-x_1(t) + x_2^3(t) - 3x_3(t)x_4(t) + d(t)),\\
        x_2(t+1)=x_2(t) + \tau_s(-x_1(t) - x_2^3(t)),\\
        x_3(t+1)=x_3(t) + \tau_s(x_1(t)x_4(t)-x_3(t)),\\
        x_4(t+1)=x_4(t) + \tau_s(x_1(t)x_3(t)-x_4^3(t)),\\
        \end{cases}
    \end{equation*}
    with $\tau_s=0.01$ and the safe set $\mathbb{X}=\{\,(x_1,x_2,x_3,x_4)^{\top}\mid \sum_{i=1}^4 x_i^2 \leq 1\,\}$. The disturbance $d(t)$ follows the uniform distribution on the interval $[-1,1]$.
\end{example} 

\begin{example}
\label{ex:lorenz}
    Consider a 6-dimensional Lorenz model adapted from \cite{lorenz1996predictability}:
\begin{equation*}
    \begin{cases}
        x_1(t+1) = x_1(t) + \tau_s\big((x_2(t) - x_5(t))x_6(t) - x_1(t) + d(t)\big),\\
        x_2(t+1) = x_2(t) + \tau_s\big((x_3(t) - x_6(t))x_1(t) - x_2(t) \big),\\
        x_3(t+1) = x_3(t) + \tau_s\big((x_4(t) - x_1(t))x_2(t) - x_3(t) \big),\\
        x_4(t+1) = x_4(t) + \tau_s\big((x_5(t) - x_2(t))x_3(t) - x_4(t) \big),\\
        x_5(t+1) = x_5(t) + \tau_s\big((x_6(t) - x_3(t))x_4(t) - x_5(t) \big),\\
        x_6(t+1) = x_6(t) + \tau_s\big((x_1(t) - x_4(t))x_5(t) - x_6(t) \big),\\
\end{cases}
\end{equation*}
with $\tau_s = 0.01$. The disturbance $d(t)$ follows a truncated normal distribution, defined on $[-3,3]$ with a mean of zero and a standard deviation of $0.5$. The safe set is defined as $\mathbb{X}=\{\,(x_1,\ldots,x_6)^{\top}\mid \sum_{i=1}^6 x_i^2 \leq 1\,\}$.
\end{example}

\begin{table}[H]
\caption{Parameter Settings for Data-Driven SBFs and Distribution-Known Methods.}
\setlength{\tabcolsep}{2.5mm}{
\begin{tabular}{ccccccc}
\hline
Examples & Degree & $\tau$ & $U_{\bm{a}}$  \\ \hline
\ref{ex:arch4}       & 10 & 0.01  & 10 \\
\ref{ex:vinc}        & 10 & 0.01  & 10 \\
\ref{ex:bc4}         & 10 & 0.01  & 10 \\
\ref{ex:stable}      & 6  & 0.01  & 10 \\
\ref{ex:3dvanderpol} & 6  & 0.01  & 10 \\
\ref{ex:sank}        & 4  & 0.01  & 15 \\
\ref{ex:lorenz}      & 4  & 0.01  & 5  \\
\hline
\end{tabular}}
\label{tab:para}
\end{table}

For the Data-Driven RBFs and Robust-Set methods, the polynomial degree of $h(\bm{a}, \bm{x})$ is set to 4 in Examples \ref{ex:arch4} - \ref{ex:3dvanderpol}, and to 2 in Examples \ref{ex:sank} and \ref{ex:lorenz}. The parameters used in the Data-Driven SBFs and Distribution-Known methods are summarized in Table~\ref{tab:para}.

\end{document}